\documentclass[11pt]{article}

\usepackage{ntheorem}
\usepackage{amsmath}
\usepackage{latexsym}
\usepackage{xspace}
\usepackage{amssymb}
\usepackage{fancybox}
\usepackage{graphicx}
\usepackage{algorithmic}
\usepackage{algorithm}
\usepackage{nicefrac}

\newcommand{\sfrac}[2]{#1\nicefrac{}{#2}}

\makeatletter
\floatstyle{ruled}
\newfloat{fragment}{H}{lop}
\floatname{fragment}{Algorithm}
\renewcommand{\floatc@ruled}[2]{\vspace{2pt}{\@fs@cfont #1.\:} #2 \par
 \vspace{1pt}}
\makeatother

\theoremsymbol{\ensuremath{_\Box}}
\theoremstyle{break}
\theoremheaderfont{\scshape}
\theorembodyfont{\slshape}
\newtheorem{Thm}{Theorem}
\newtheorem{theorem}[Thm]{Theorem}
\newtheorem{Lem}[Thm]{Lemma}

\newtheorem{Cor}[Thm]{Corollary}

\newtheorem{claim}[Thm]{Claim}
\theoremstyle{plain}
\theorembodyfont{\rmfamily}

\theorembodyfont{\itshape}
\newtheorem{Def}{Definition}
\newenvironment{proof}{\noindent {\sc Proof:}}{$\Box$ %\medskip
        }
 {
        \begin{enumerate}}{\end{enumerate}}

\newcommand{\marginlabel}[1]%
{\mbox{}\marginpar{\it{\raggedleft\hspace{0pt}#1}}}
\newcommand\card[1]{\left| #1 \right|} %cardinality of set S; usage \card{S}
\newcommand\set[1]{\left\{#1\right\}} %usage \set{1,2,3,,}
\newcommand\poly{\mbox{poly}}  %usage \poly(n)

\newcommand\R{\mathbb{R}}

\newcommand{\row}[2]{#1^{(#2)}}
\newcommand{\mat}[3]{#1^{(#2)}_{#3}}
\newcommand{\norm}[1]{\lVert #1 \rVert}

\newcommand{\tilT}{\tilde{T}}

\newcommand{\E}{\mathop{\mathbb E}\displaylimits}
\newcommand{\Exp}{\mathop{\mathbb E}\displaylimits}
\newcommand{\EExp}{\mathop{\hat{\mathbb E}}\displaylimits}
\newcommand{\Var}{\mathop{\mathbb{V}}\displaylimits}

\newcommand{\bias}{\hat{B}}

\def\shownotes{1}  %set 1 to show author notes
\ifnum\shownotes=1
\newcommand{\authnote}[2]{{ $\ll$\textsf{\footnotesize #1 notes: #2}$\gg$}}
\else
\newcommand{\authnote}[2]{}
\fi

\newcommand{\Tnote}[1]{{\authnote{Tengyu}{#1}}}

\title{More algorithms for provable dictionary learning}

\author{Sanjeev Arora\thanks{Princeton University, Computer Science Department and Center for Computational Intractability. Email: arora@cs.princeton.edu. This work is supported by the NSF grants CCF-0832797, CCF-1117309, CCF-1302518, DMS-1317308, and Simons Investigator Grant.} \and Aditya Bhaskara\thanks{Google Research NYC.  Email: bhaskara@cs.princeton.edu.  Part of this work was done while the author was a Postdoc at EPFL, Switzerland.}  \and Rong Ge\thanks{Microsoft Research. Email: rongge@microsoft.com. Part of this work was done while the author was a graduate student at Princeton University and was supported in part by NSF grants CCF-0832797, CCF-1117309, CCF-1302518, DMS-1317308, and Simons Investigator Grant.} \and Tengyu Ma\thanks{Princeton University, Computer Science Department and Center for Computational Intractability. Email: tengyu@cs.princeton.edu. This work is supported by the NSF grants CCF-0832797, CCF-1117309, CCF-1302518, DMS-1317308, and Simons Investigator Grant.}}
\begin{document}

\maketitle
\begin{abstract}
In {\em dictionary learning}, also known as {\em sparse coding}, the algorithm is given samples of the form $y = Ax$ where $x\in \mathbb{R}^m$ is an unknown random sparse vector and $A$ is an unknown dictionary matrix in $\R^{n\times m}$ (usually $m > n$, which is the {\em overcomplete} case). The goal is to learn $A$ and $x$. 
This problem has been studied in neuroscience, machine learning, visions, and image processing. In practice it is solved by heuristic algorithms and provable algorithms seemed 
hard to find. Recently, provable algorithms were found that work if the unknown feature vector $x$ is $\sqrt{n}$-sparse or even sparser. Spielman et al.~\cite{DBLP:journals/jmlr/SpielmanWW12} 
did this for dictionaries where $m=n$; Arora et al.~\cite{AGM} gave an algorithm for overcomplete ($m >n$) and incoherent matrices $A$; and
Agarwal et al.~\cite{DBLP:journals/corr/AgarwalAN13} handled a similar case but with weaker guarantees.

This raised the problem of designing provable algorithms that allow sparsity $\gg \sqrt{n}$ in the hidden vector $x$. The current paper designs algorithms that allow sparsity up to $n/poly(\log n)$.
It works for a class of matrices where features are {\em individually recoverable}, a new notion identified in this paper that may motivate further work. 

The algorithm runs in quasipolynomial time because they use limited enumeration.

\end{abstract}
\newpage

\section{Introduction}

Dictionary learning, also known as {\em sparse coding}, tries to understand the structure of observed samples $y$ by representing them as sparse linear combinations of ``dictionary'' elements. More precisely, there is an unknown dictionary matrix $A \in \R^{n\times m}$ (usually $m > n$, which is the {\em overcomplete} case), and the algorithm is given samples $y = Ax$ where $x$ is an unknown random sparse vector. (We say a vector is {\em $k$-sparse}  if it has at most $k$ nonzero coordinates.) The goal is to learn $A$ and $x$.
Such sparse representation was first studied in neuroscience, where Olshausen and Field~\cite{olshausen1997sparse} suggested that dictionaries fitted to real-life images have similar properties as the receptive fields of neurons in the first layer of visual cortex. Inspired by this neural analog, dictionary learning is widely used in machine learning for {\em feature selection}~\cite{DBLP:conf/nips/ArgyriouEP06}. More recently the idea of sparse coding has also influenced deep learning~\cite{boureau2007sparse}. In image processing, learned dictionaries have been successfully applied to image denoising~\cite{elad2006image}, edge detection~\cite{mairal2008discriminative} and super-resolution~\cite{yang2008image}. 

Provable guarantees for dictionary learning have seemed difficult because the obvious math programming formulation is nonconvex: both $A$ and the $x$'s are unknown. 
Even when the dictionary $A$ is known, it is in general NP-hard to get the sparse combination $x$ given {\em worst-case} $y$~\cite{davis1997adaptive}. This problem of decoding $x$ given $Ax$ with full knowledge of $A$ is called {\em sparse recovery} or {\em sparse regression}, and is closely related to {\em compressed sensing}. For many types of dictionary $A$, sparse recovery was shown to be tractable even on worst-case $y$, starting with such a result for {\em incoherent} matrices  by Donoho and Huo~\cite{donoho2001uncertainty}. However in most early works $x$ was constrained to be $\sqrt{n}$-sparse, until Candes, Romberg and Tao~\cite{candes2006robust} showed how to do sparse recovery even when the sparsity is $\Omega(n)$, assuming $A$ satisfies the {\em restricted isometry property} (RIP) (which random matrices do).

But dictionary learning itself (recovering $A$ given samples $y$) has proved much harder and heuristic algorithms are widely used. Lewicki and Sejnowski~\cite{lewicki2000learning} designed the first one, which was followed by the method of optimal directions (MOD)~\cite{engan1999method} and K-SVD~\cite{aharon2006img}. See~\cite{AharonThesis} for more references. However, until recently there were no algorithms that provably recovers the correct dictionary. Recently Spielman et al.~\cite{DBLP:journals/jmlr/SpielmanWW12} gave such an algorithm  for the
full rank case (i.e., $m=n$) and the unknown feature vector $x$ is $\sqrt{n}$-sparse. However, in practice overcomplete dictionaries ($m > n$) are
preferred.  Arora et al.~\cite{AGM} gave the first provable learning algorithm for overcomplete dictionaries that runs in polynomial-time; they required $x$ to be $n^{1/2 -\epsilon}$-sparse (roughly speaking)
and $A$ to be incoherent. Independently, Agarwal et al.~\cite{DBLP:journals/corr/AgarwalAN13} gave a weaker algorithm that also assumes $A$ is incoherent and allows $x$ to be $n^{1/4}$-sparse. 
Thus all three of these recent algorithms cannot handle sparsity more than $\sqrt{n}$, and this is a fundamental limitation of the technique: they require two random $x, x'$  to
intersect in no more than $O(1)$ coordinates with high probability, which fails to hold when sparsity $\gg \sqrt{n}$. Since sparse recovery (where $A$ is known) is possible even up to
sparsity $\Omega(n)$, this raised the question whether dictionary learning is possible in that regime. In this paper we will refer to feature vectors with sparsity $n/poly(\log n)$
 as {\em slightly-sparse}, since methods in this paper do not seem to allow density higher than that. 
%for any matrices which allow spa
 
 In our recent paper on deep learning~(\cite{DBLP:journals/corr/AroraBGM13}, Section 7) we showed how to solve dictionary learning in this regime for dictionaries which are adjacency matrices of 
  random weighted sparse bipartite graphs; these are known to allow sparse recovery albeit with a slight twist in the problem definition~\cite{DBLP:conf/soda/Indyk08, DBLP:journals/tit/JafarpourXHC09, IndykStrauss}.
  Since real-life dictionaries are probably
  not random, this raises the question whether dictionary learning is possible in the slightly sparse case for other dictionaries.
  The current paper gives quasipolynomial-time algorithms  for learning more such dictionaries. The running time is quasipolynomial time because it 
  uses limited enumeration (similarly, e.g., to algorithms for learning gaussian mixtures). Now we discuss this class of dictionaries. 
  
Some of our discussion below refers to nonnegative dictionary learning, which constrains matrices $A$ and hidden vector $x$ to have nonnegative entries. This is a popular variant proposed by Hoyer~\cite{hoyer2002non},  motivated again partly by the neural analogy. Algorithms like NN-K-SVD~\cite{aharon2005k} were then applied to image classification tasks. This version is also related to {\em nonnegative matrix factorization}~\cite{lee1999learning}, which has been observed to lead to factorizations that are usually sparser and more local than traditional methods like SVD.

 \subsection{How to define dictionary learning?}
 
 Now we discuss what versions of dictionary learning make more sense than others. For exposition purposes we refer to the coordinates of the hidden vector $x$ as
 {\em features}, and those of the visible vector $y=Ax$ as {\em pixels}, even though the discussion applies to more than just computer vision.
 % in mind.  
 %In applications, the hidden vector $X$ corresponds to {\em features} that are present in the datapoint, and their relative amounts.
 %Thus 
 Dictionary learning as defined here ---which is the standard definition---assumes that features' effect on the pixels add linearly. % additive. 
 %For simplicity we will refer to coordinates in the input as \textquotedblleft pixels,\textquotedblright even though the discussion below applies to 
 %more than just computer vision. 
 
But, the problem definition is somewhat arbitrary. On the one hand one could consider
 {\em more general} (and nonlinear) versions of this problem ---for instance in 
 vision,  dictionary learning is part of a system that has to deal with {\em occlusions} among objects that may hide part of a 
 feature, and to incorporate the fact that features may be present with an arbitrary translation/rotation. On the other hand, 
 one could consider more specific versions that place restrictions on the dictionary, since not all 
 dictionaries may make sense in applications. We consider this latter possibility now, with the usual caveat that
 it is nontrivial to cleanly formalize properties of real-life instances. 
 
 One reasonable property of real-life dictionaries is that each feature does not involve most pixels.  This implies that column vectors of $A$ are
 {\em relatively sparse}. Thus  matrices with RIP property ---at least if they are dense--- do not seem a good match\footnote{By contrast, in the usual setting
 of compressed sensing, the matrix provides a basis for making measurements, and its density is a nonissue.}.

 Another intuitive property is that features are {\em individually recoverable}, which means, roughly speaking, that to an observer who knows the dictionary, the presence of
 a particular feature should not be confusable with the effects produced by the usual distribution of other features
 (this is an average-case condition, since $x$ satisfies stochastic assumptions). In particular, one should be able to detect its presence
 by looking only at the pixels it would affect. 
 
 Thus it becomes clear that not all matrices that allow sparse recovery are of equal interests. The paper of
 Arora et al.~\cite{AGM} restricts attention to {\em incoherent} matrices, where the columns have pairwise inner product at most
 $\mu/\sqrt{n}$ where $\mu$ is small, like $\poly(\log n)$. These make sense on both the above counts. First, they can have fairly sparse
 columns. Secondly they satisfy $A^T A \approx I$, so given $Ax$ one can take its inner product with the $i$th column $A_{i}$ to roughly determine the extent to which feature
 $i$ is present.  But incoherent matrices restrict sparsity to $O(\sqrt{n})$, so one is tempted by RIP matrices but, as mentioned, their columns are fairly dense. 
 Furthermore, RIP matrices were designed to allow sparse recovery for worst-case feature vectors whereas in dictionary learning these are stochastic.
As mentioned, sparse random graphs (with random edges weights in $[-1,1]$) check all the right boxes (and were handled in our recent paper on deep learning) but require positing
that the dictionary has no structure. The goal in the current paper is to move beyond random graphs.
 
 \paragraph{Dictionaries with individually recoverable features.} Let us try to formulate the property that features are {\em individually recoverable}. We hope this definition and discussion will 
 stimulate further work (similar, we hope, to Dasgupta's formalization of  separability for gaussian mixtures~\cite{DBLP:conf/focs/Dasgupta99}). Let us assume  that the coordinates of $x$ are pairwise independent. Then the presence of the $i$th feature (i.e., $x_i \neq 0$) changes the 
 conditional distribution of those pixels involved in $A_{i}$, the $i$th column of $A$.  Features are said to be {\em individually recoverable} if this change in conditional distribution is
 not obtainable from other combinations of features that arise with reasonable probability.  
 This statistical property is hard to work with and below we suggest some (possibly too strong) combinatorial properties of the support of $A$ that imply it.
 Better formalizations seem quite plausible and are left for future work.

\section{Definitions and Results}
\label{sec:prelim}

The dictionary is an unknown matrix $A \in \R^{n\times m}$. We are given i.i.d samples $y$ that are generated by $y = Ax$, where $x \in \R^m$ is chosen from some distribution.
We have $N$ samples $y^i = Ax^i$ for $i = 1,\dots, N$.  
As in the introduction,
we will refer to coordinates of $x$ as {\em features} and those of $y$ as {\em pixels}, even though vision isn't the only intended application. 
%\Rnote{I don't think any of the latter sections use the ``features'' and ``pixels'', they are called ``hidden nodes'' and ``observed nodes'' defined below. Do we want to change the names?}
For most of the paper we assume the entries of $x$ are independent Bernoulli variables: each $x_i$ is 1 with probability $\rho$ and $0$ with probability $1-\rho$; we refer to this  as $\rho$-Bernoulli distribution. 
This assumption can be relaxed somewhat: we only need that entries of $x$ are pairwise independent, and $e^Tx$ should satisfy concentration bounds for reasonable vectors $e$. 
The nonzero entries of $x$ can also be in $[1,c]$ instead of being exactly 1.

The $j$th column of $A$ is denoted by  $A_j$, the $i$th row of $A$ by $\row{A}{i}$, and the entries of $A$ are denoted by $\mat{A}{i}{j}$. 

For ease of exposition we first describe our learning algorithm for nonnegative dictionaries (i.e., $\mat{A}{i}{j} \geq 0$, for all $i\in [n], j\in [m]$) and then in Section~\ref{sec:general} describe the generalization to
the general case. Note that the subcase of nonnegative dictionaries is also of practical interest.
%will apply to 

\subsection{Nonnegative dictionaries} 
By normalizing, we can assume without loss of generality that the expected value of each pixel is $1$, that is, $\Exp[y_i] = \E[(Ax)_i] = 1$. 
We also assume that  $|\mat{A}{i}{j}| \leq \Lambda$ for some constant $\Lambda$ \footnote{Though the absolute value notations are redundent in the nonnegative case, we keep them so that they can be adapted to the general case. Similarly for the definition of $G_b$ later.}: no entry of $A$ is too large. 
Let $G_{b}$ be the bipartite graph defined by entries of $A$ that have magnitudes larger than or equal to $b$, that is, $G_b = \{(i,j) : |\mat{A}{i}{j}| \ge b, i\in [n],j\in [m]\}$. %Using the idea of viewing $y = Ax$ as a bipartite graph, we call coordinates of $x$ the hidden nodes and coordinates of $y$ the observed nodes.
We make two assumptions about this graph (the parameters $d$, $\sigma$ and $\kappa$ will be chosen later).
%Let $\mathcal{Y}$ denote the set of all samples $\mathcal{Y} = \set{y^1,\dots, y^N}$

%\Tnote{ Note that we are replacing $\sigma^2/\tau$ by $\tau$, and $\sigma^2d/\kappa$ by $\kappa$, where it was that $\tau = \Omega(\log n)$ and $\kappa = \Omega(\log^2 n)$, and the constant hidden in $\Omega$ presumably depends on $\Lambda$}
\begin{description}
\item[Assumption 1:] (Every feature has significant effect on pixels) There are at least $d$ edges with weights larger than $\sigma$ for every feature $j$. That is, the degree of $G_{\sigma}$ on the feature side is always larger than $d$. 
%\begin{description}
\item[Assumption 2:] (Low pairwise intersections among features) In $G_{\tau}$ the neighborhood of each feature (that is, $\{i\in [n] : \mat{A}{i}{j} \ge \tau \}$) has intersection up to $d/10$ (with total weight $< d\sigma/10$) with each of at most $o(1/\sqrt{\rho})$ other features, and intersection at most $\kappa$ with the neighborhood of each remaining features. Here $\tau$ is 
%$O(\sigma^4/\log n)$ 
$O_{\theta}(1/\log n)$ as explained below and $\kappa = O_{\theta}(d/\log^2 n)$.
\end{description}

\noindent{\bf Guideline through notation:} We will think of $\sigma \le 1$ as a small constant, $\Lambda \ge 1$ a constant, and $\Delta$ a sufficiently large constant which is used to control the assumption. %the running time and sample complexity as described below. 
Let $\theta = (\sigma, \Lambda, \Delta)$ and we use the notation $O_{\theta}(\cdot)$ to hide the dependencies of $\sigma,\Lambda,\Delta$. Also, we think of $m$ as not much larger than $n$, $\rho < 1/\poly(\log n)$, and $d \ll n$.
 The normalization assumption implies (for all practical purposes in the algorithm) that $m d \rho \in [n/\Lambda, n/\tau]$. 
 We typically think of $d$ as $1/\rho$, hence
 % Thus $d$ is roughly $1/\rho > poly(\log n)$ and 
 a running time
 of $m^d$ would be bad (though it is unclear a priori how to even achieve that). 
 
 Precisely, for our algorithms to work, 
 we need $d \ge \Delta\Lambda \log^2 n/\sigma^2$, $\tau = O(\sigma^4/\Delta \Lambda^2 \log n) = O_\theta(1/\log n)$ and 
$\kappa = O(\sfrac{\sigma^8 d}{\log^2 n \Delta^2 \Lambda^6}) = O_{\theta}(d/\log^2 n)$ for some sufficiently large constant $\Delta$ and the density $\rho = o(\sfrac{\sigma^5}{\Lambda^{6.5}\log^{2.5} n}) = o_{\theta}(1/\log^{2.5} n)$.
Note that if $G_{\tau}$ were like a random graph (i.e. if features affect random
sets of $d$ pixels) when $d^2 \ll n$, the pairwise intersection $\kappa$ between the neighborhoods of two features in $G_{\tau}$ would be $O(1)$. However, we allow these intersections to be $\kappa = O_{\theta}(d/\log^2 n)$. 

Now we give a stronger version of Assumption 2 which will allow a stronger algorithm.

\begin{description}
\item[Assumption 2':]  In $G_{\tau}$, the pairwise intersection of the neighborhoods of any two features $j,k$ is less than $\kappa$, where $\tau = O_{\theta}(1/\log n)$ and $\kappa = O_{\theta}(d/\log^2 n)$.
\end{description}

The algorithm can only learn the real-valued matrix approximately.  Two dictionaries are {\em close} if they satisfy the following definition:

\begin{Def}[$\epsilon$-equivalent] Two dictionaries $A$ and $\hat{A} \in R^{n\times m}$ are $\epsilon$-equivalent, if for a random vector $x\in \R^m$ with independent $\rho$-Bernoulli components, with high probability $Ax$ and $\hat{A}x$ are entry-wise $\epsilon$-close.
\end{Def}

\iffalse
\Tnote{Comment the following after check.. }

\Tnote{Lemma~\ref{lem:exists_sign} needs $\tau = O(\sfrac{\sigma^2}{\Delta\log n})$ }

\Tnote{Lemma~\ref{lem:expandsignature} needs $\tau = o(1)$ and $\kappa = o(d)$}

\Tnote{Needs $\tau = O(\sfrac{d\sigma^2}{\Delta\Lambda^2\log n})$, $\kappa = O(\sfrac{\sigma^4d}{\log^2 n \Delta^2\Lambda^6})$}

\Tnote{Claim~\ref{claim:biasmaxclose} needs $\rho \le c/\log^2 n$ for a small enough constant $c$}

\fi

\begin{theorem}[Nonneg Case]
\label{thm:nonneg:main}
Under Assumptions 1 and 2, when $\rho = o(\sigma^5/\Lambda^{6.5}\log^{2.5} n) = o_{\theta}(1/\log^{2.5} n)$ , Algorithm~\ref{alg:main} runs in $n^{O(\Lambda \log^2 n/\sigma^4)}$ time, uses $\poly(n)$ samples and outputs a matrix that is $o(\rho)$-equivalent to the true dictionary $A$. Furthermore, under Assumptions 1 and 2' the same algorithm returns a dictionary that is  $n^{-C}$-equivalent to the true dictionary, while using $n^{4C+3}$ samples, where $C$ is a large constant depending on $\Delta$. \footnote{Recall that $\Delta$ is a sufficiently large constant that controls the parameters of the assumptions.}
%If the support of entries that are larger than $???$ in matrix $A$ form a random bipartite graph \Rnote{Actually we should only need very few deterministic properties}, then given samples of the form $y = %Ax$ where entries of unknown vector $x$ comes from independent $\rho$-Bernoulli distributions there is 
%an algorithm running in $n^{O()}$ time, using $n^?$ samples\Rnote{I'm not really sure but hopefully sample complexity is still polynomial?}, and with high probability returns a matrix $\hat{A}$ that is %$\epsilon$-equivalent to $A$.
\end{theorem}

The theorem is proved in Section~\ref{sec:nonneg}. 

\noindent{\em Remark:} Assumption 2 is {\em existentially} close to optimal, in the sense that if it is significantly violated: e.g., if there are $poly(1/\rho)$ features that intersect 
the neighborhood of feature $j$ using edges of total weight  $\Omega(\mbox{$\ell_1$-norm of $A_j$})$ then feature $j$ is no longer individually recoverable: its effect can be duplicated whp by combinations of these other features.
% Ideally one would allow the pairwise intersection to be as large as, say, $d/2$. Or one 
But a more precise characterization of individual recoverability would be nice, as well as a matching algorithm.
% (probably using the amount of {\em weight} the features
%assign to their common intersection) would be nice. 
%Also, in this case we would only be able to get $o(1)$-equivalent dictionaries, and cannot make the dictionaries arbitrarily close.%The problem is that such matrices may in general not allow individual
%recovery of features, the property discussed in the introduction. Extending our methods to  these other settings is left for future work. 

%\Snote{Current writeup is using this}

\subsection{Dictionaries with negative entries} 

When the edges can be both positive and negative, it is no longer valid to assume the expectation of $y_i$'s are equal to $1$. Instead, we choose a different normalization: the 
variances of $y_i$'s are 1. We still assume magnitude of edge weights are at most $\Lambda$,  and that features don't overlap a lot as described in Assumption 1 and 2. 
We also need one more assumption to bound the variance contributed by the small entries. 

\begin{description}
\item[Assumption G1:] The degree of $G_{\sigma}$ on side of $x$ is always larger than $2d$. 
\item[Assumption G2':] 
 In $G_{\tau}$, the pairwise intersection of the neighborhoods of any two features $j,k$ is less than
%In $G_{\tau}$ the neighborhood of each feature has %intersection up to $d/10$ (with total weight $< d\sigma/10$) with each of at most $o(1/\sqrt{\rho})$ other features, and intersection at most 
$\kappa$,% with the neighborhood of each remaining features,
 where $\tau = O_\theta(1/\log n)$  and $\kappa = O_\theta(d/\log^2 n)$.

%\tau = \sigma^4/\log n

%\Tnote{Should we use 2 or 2' here? It was originally 2' as stated in the theorem below.}

\item[Assumption G3:] (small entries of $A$ don't cause large effects) $\rho ||A^{(i)}_{\le \tau}||_2^2 \le \gamma$, where $A^{(i)}_{\le \delta}$ be the vector that only contains the entries of $A^{(i)}$ that are at most $\delta$, and $\gamma =  \sfrac{\sigma^4}{2\Delta\Lambda^2\log n}$. 
%Let $A^{(i)}_{\le \delta}$ be the vector that only contains the entries of $A^{(i)}$ that are at most $\delta$. When $\delta = f(\gamma)$, $\rho ||A^{(i)}_{\le \delta}||_2^2 \le \gamma$.
\end{description}
Note that Assumption G1 differs from Assumption 1 by a constant factor 2 just to simplify some notations later. Assumption G2' is the same as before. 

This assumption G3 intuitively says that for each $y_i  = \sum_k \mat{A}{i}{k} x_k$, the smaller $\mat{A}{i}{k}$'s should not contribute too much to the variance of $y_i$. This is automatically satisfied for nonnegative dictionaries because there can be no cancellations. Notice that this assumption is talking about rows of matrix $A$ (corresponding to pixels), whereas the earlier assumptions talk about columns of $A$ (corresponding to features). Also, consider $\tau$ to be the smallest number between what is required by Assumption G2' and G3.

%In this case, 
In term of parameters, we still need $d \ge \Delta\Lambda \log^2 n/\sigma^2$, and 
$\kappa = O(\sigma^8 d/\log^2 n \Delta^2 \Lambda^6) $$= O(d/\log^2 n)$. 
%However, we shall choose $\gamma = \sigma/\Lambda \sqrt{\Delta \log n}$ and $\tau = f(\gamma)$. 
As before $\Delta$ is a large enough constant. 

%{\em Write assumptions here and theorem statement.}
%
%The parameters $d$, $\kappa$ and $\tau$ are chosen as \Rnote{need to do}

\begin{theorem}
\label{thm:general:main}
Under Assumptions G1, G2' and G3, when $\rho = o(\sigma^5/\Lambda^{6.5}\log^{2.5} n) = o_{\theta}(1/\log^{2.5} n)$ there is an algorithm that runs in $n^{O(\Delta\Lambda \log^2 n/\sigma^2)}$ time, uses $n^{4C+5}m$ samples  and outputs a matrix that is $n^{-C}$-equivalent to the true dictionary $A$, where $C$ is a constant depending on $\Delta$. 
\end{theorem}

The algorithm and the proof of Theorem~\ref{thm:general:main} are sketched in Section~\ref{sec:general}.

%\Rnote{The assumptions to make}

%1. There are at least $d$ edges larger than $\sigma$ (defined later). 2. For two hidden nodes, the number of observed nodes that intersect both is bounded by $\sigma^2 d/5C^3\log n$

\section{Nonnegative Dictionary Learning}\label{sec:nonneg}

Recall that dictionary learning seems hard because both $A$ and $x$ are unknown.
To get around this problem, previous works (e.g. \cite{AGM}) try to extract information about the assignment $x$ without first learning $A$ (but assuming nice properties of $A$). After finding $x$, recovering $A$ becomes easy. In \cite{AGM} the unknown $x$'s were recovered via an overlapping
clustering procedure. The procedure relies on incoherence of $A$, as when $A$ is incoherent it is possible to test whether the support of $x^1$, $x^2$ intersect.
% namely, inner product of two samples, $\iprod{y^1}{y^2} =\iprod{Ax^1}{Ax^2}$ exceeds a certain threshold iff the support of $x^1, x^2$ have a nonempty intersection.
This idea fails when $x$ is only slightly sparse, because in this setting the supports of $x^1, x^2$ always have a large intersection.

Our algorithm here relies on correlation among {\em pixels}. The key observation is: if the $j$th bit in $x$ is 1, then $A x = A_j + \sum_{k\ne j} A_k x_k$. Pixels with high values in $A_j$ tend to be {\em elevated} above their mean values (recall $A$ is nonnegative). At first it is unclear how
this simultaneous elevation can be spotted, since $A_j$ is unknown and these elevations/correlations among pixels are
much smaller than the standard deviation of individual pixels. %affects only a few pixels. However, because pairwise intersection among features is assumed to be small, it can be shown that
Therefore we look for local regions ---small subsets of pixels--- in $A_j$ where this effect is  significant in the aggregate (i.e., sum of pixel values),
 and can be used to consistently predict the value of $x_j$. 
These are called the {\em signature sets} (see  Definition~\ref{def:signature_set}). If we can identify signature sets, they can  give us a good estimation of whether the feature $x_j$
is present.

Since the signature sets are small, in quasi-polynomial time we can afford to enumerate all sets of that size, and check if the pixels in these sets are likely to be elevated together. However, this does not solve the problem, because there can be many sets ---called {\em correlated sets} below--- that show similar {\em correlations} and look similar to signature sets. % but are not signature sets. 
It is hard to separate signature sets from other correlated sets when the size of the set is small.
%Such signature sets can be discovered in quasipolynomial time by exhaustive enumeration of all subsets of a small number of pixels ---this is similar to exhaustive enumeration used for example in approximation algorithms for dense graphs~\cite{}. 
%But a big problem remains: this exhaustive enumeration throws up many fake candidates (the correlation sets of Definition~\ref{def:correlated_set}) 
%which look just like signature sets but do not actually help to predict any hidden features. 
This leads to the next idea: try to {\em expand} a signature set by first estimating the corresponding column of $A$, and then picking large entries in that column. % into a  for a column of $A$. 
 The resulting sets are called %We call the results of expanding signature sets 
{\em expanded signature sets}; these have size $d$ (and hence could not have been found by exhaustive guessing alone).
If the set being expanded is indeed a signature set, this expansion process can correctly estimate the column of $A$. % the observed distribution on pixels so well that it {\em has} to correspond to
 %a hidden feature.  
 We give algorithms that can find expanded signature sets, and using these sets we can get a rough estimation for the matrix $A$. Finally, we also give a procedure that leverages the {\em individually recoverable} properties of the features, and refines the solution to be inverse polynomially equivalent to the true dictionary.
%That serves as a validation that the guessed set was a true signature set, and also allows the guess for the column of $A$ to be
%refined to a much higher accuracy. 

The high level algorithm is described in Algorithm~\ref{alg:main} (the concepts such as correlated sets, and empirical bias are defined later).
To simplify the proof the algorithm description uses Assumption 2'; we summarize later (in Section~\ref{subsec:assumption2}) what changes with Assumption 2. 

The main algorithm has three main steps. Section~\ref{subsec:signaturesets} explains how to test for correlated sets and expand a set (1-2 in Algorithm~\ref{alg:main}); Section~\ref{subsec:expandedsignatureset} shows how to find expanded signature sets and a rough estimation of $A$ (3-6 in Algorithm~\ref{alg:main}); finally Section~\ref{subsec:refine} shows how to refine the solution and get $\hat{A}$ that is inverse polynomially equivalent to $A$ (7-10 in Algorithm~\ref{alg:main}).

%The goal of the algorithm is to learn an $\epsilon$-equivalent dictionary $\hat{A}$.

\subsection{Correlated Sets, Signature Sets and Expanded Sets}

\label{subsec:signaturesets}

%For a fixed subset $T\subset [n]$ of pixels, we are interested in whether they are all connected to the same  node. In order to do this, the algorithm looks at the sum of all the coordinates of $y_i (i\in T)$. 

We consider a set $T$ of size $t = \Omega(\poly\log n)$ (to be specified later), and denote by $\beta_T$ the random variable representing the
 sum of all pixels in $T$, i.e., $\beta_T = \sum_{i\in T} y_i$. We can expand $\beta_T$ as 

\[\beta_T = \sum_{i\in T} y_i = \sum_{i\in T} \left(\sum_{j=1}^m A_{j}^{(i)}x_j\right) = \sum_{j=1}^m \left(\sum_{i\in T} \mat{A}{i}{j}\right) x_j.\]

Let $\beta_{j,T} =  \left(\sum_{i\in T} \mat{A}{i}{j}\right)$ be the contribution of $x_j$ to the sum $\beta_T$, then $\beta_T$ is just 

\begin{equation}
\beta_T = \sum_{j=1}^m \beta_{j,T}x_j \label{eqn:1}
\end{equation}

Note that by the normalization of $\E[y_i]$, we have $\Exp[\beta_T] = \sum_{i\in T} \Exp[y_i] = t$. Intuitively, if for all $j$, $\beta_{j,T}$'s are relatively small, $\beta_T$ should concentrate around its mean. On the other hand, if there is some $j$ whose coefficient $\beta_{j,T}$ is significantly larger than other $\beta_{k,T}$, then $\beta_T$ will be {\em elevated} by $\beta_{j,T}$ precisely when $x_j = 1$. That is, with probability roughly $\rho$ (corresponding to when $x_j=1$), we should observe $\beta_T$ to be roughly $\beta_{j,T}$ larger than its expectation. 

Now we make this precise by defining such set $T$ with only one large coefficient $\beta_{k,T}$ as signature sets. 
%be made formal.
%as $\beta_T= \sum_{i\in T, j\in [m]} \mat{A}{i}{j} x_j$. If all the nodes in $T$ are connected to the same feature $j$, then for samples with $x_j = 1$ the sum $S_T$ should be significantly larger than samples with $x_j = 0$. On the other hand, if there is no feature that connects to many nodes in $T$, then by concentration bounds we would expect it is very unlikely for $S_T$ to be much larger than its mean.

\begin{Def}[Signature Set] \label{def:signature_set} A set $T$ of size $t$ is  a {\em signature set} for $x_j$, if $\beta_{j,T} \ge \sigma t$, and for all $k\ne j$, the contribution $\beta_{k,T} \le \sfrac{\sigma^2t}{\Delta\log n}$. Here $\Delta$ is a large enough constant.
\end{Def}

The following lemma formalizes the earlier intuition that if $T$ is a signature set for $x_j$, then a large $\beta_T$ is highly correlated with the event $x_j = 1$.

\begin{Lem}\label{lem:symmetric_diff} 
Suppose $T$ of size $t$ is a signature set for $x_j$ with $t =\omega(\sqrt{\log n})$. Let $E_1$ be the event that $x_j =1$ and $E_2$ be the event that $\beta_T \ge \Exp[\beta_T]+ 0.9\sigma t$. Then for large constant $C$ (depending on the $\Delta$ in Definition~\ref{def:signature_set})
\begin{enumerate}
\item $\Pr[E_1] + n^{-2C}\ge \Pr[E_2] \ge \Pr[E_1] - n^{-2C}$. 
\item $\Pr[E_2|E_1] \ge 1 - n^{-2C}$, and $\Pr[E_2|E_1^c] \le n^{-2C}$.
\item $\Pr[E_1|E_2]\ge 1 - n^{-C}$.
\end{enumerate} 
%\Tnote{or $n^{-1.5}/\rho$ , should check constant}
\end{Lem}
%\noindent{\em Remark:} The first condition can be weakened to $\Pr[E_2] \in [\Pr[E_1] \pm n^{-C}]$. In other words, $\Pr[x_j]$ can differ among the features. 

\begin{proof}
We can write $\beta_T$ as 
\begin{equation}
\beta_T = \beta_{j,T}x_j + \sum_{k\neq j}\beta_{k,T}x_k \label{eqn:2}
\end{equation}
The idea is that since $\beta_{k,T} < \sigma^2 t/(\Delta \log n)$ for all $k \ne j$, the summation in the RHS above is highly concentrated around its mean, which is actually very close to $\Exp[\beta_T] = t$. Therefore since $\beta_{j,T} > \sigma t$, we know $\beta_T > t+ 0.9 \sigma t$ essentially iff $x_j=1$. 

%Let us now prove this 
Formally, observe that $\E[ \beta_{j,T}x_j] = \rho \beta_{j,T} \le \rho \Lambda t = o(\sigma t) $, and recall that $\E[\beta_T] = t$, we have $\Exp[\sum_{k\neq j}\beta_{k,T}x_k] = (1-o(\sigma))t$. Let $M =  \sigma^2t/(\Delta\log n)$ be the upper bound for $\beta_{k,T}$, and then the variance of the sum $\sum_{k\neq j}\beta_{k,T}x_k$ is bounded by $\rho M \sum_{k\neq j}\beta_{k,T} \le Mt$. Then by calling Bernstein inequality (see Theorem~\ref{thm:bernstein_ineq}, but note that $\sigma$ there is the standard deviation), we have 
\begin{eqnarray*}
\Pr\left[\left| \sum_{k\neq j}\beta_{k,T}x_k - \Exp[\sum_{k\neq j}\beta_{k,T}x_k] \right| > \sigma t/20\right] &\le &
2\exp(-\frac{\sigma^2 t^2/400}{2Mt+ \frac{2}{3}\frac{M}{\sqrt{Mt}}\sigma t/20}) \le n^{-2C}.
%&\le& \exp(-\min\{\sigma^2t^2/(2a^2) + \frac{3}{2}t\sqrt{\theta}/a\})
%2\exp(-\frac{0.64\sigma^2t^2}{2Mt+2M\sigma t})\le n^{-1.5}
\end{eqnarray*}
where $C$ is a large constant depending $\Delta$. 

%From this, 
Part (2) immediately follows: if $x_j=1$, then $\beta_T < t+ 0.9 \sigma t$ iff the sum deviates from its expectation by more than $\sigma t/20$, which happens with probability $< n^{-2C}$.  So also if $x_j=0$, $E_2$ occurs with probability $<n^{-2C}$.

This then implies part (1), since the probability of $E_1$ is precisely $\rho$.

Combining the (1) and (2), and using Bayes' rule $\Pr[ E_1|E_2] = \Pr[ E_2 | E_1 ] \Pr[E_1] / \Pr[E_2]$, we obtain (3).
\end{proof}

Thus if we can find a signature set for $x_j$, we would roughly know the samples in which $x_j=1$.  The following lemma shows that assuming the low pairwise assumptions among
features, there {\em exists} a signature set for every feature $x_j$.
% intersection property (P2), every $x_j$ has a signature set.

\begin{Lem}\label{lem:exists_sign}
%When $t = \Omega(\Lambda^6\log^2 n/\sigma^4)$, there exists signature sets for all features $x_j$.
%Suppose for any $j,k$, the number of common neighbor of $x_j$ and $x_k$ in $G_{\sigma^2/\tau}$ is at most $\sigma^2 d/\kappa$, 
%($\alpha \ge \Omega(\Lambda\log n/\sigma^2)$), 
Suppose $A$ satisfies Assumptions 1 and 2, let $t = \Omega(\Lambda\Delta\log^2 n/\sigma^2)$, then for any $j \in [n]$,  there exists a signature set of size $t$ for node $x_j$.
\end{Lem}

\begin{proof}
We show the existence by probabilistic method. By Assumption 1, node $x_j$ has at least $d$ neighbors in $G_{\sigma}$. Let $T$ be a uniformly random set of $t$ neighbors of $x_j$ in $G_{\sigma}$. Now by the definition of $G_{\sigma}$ we have $\beta_{j,T} \ge \sigma t$. 

Using a bound on intersection size (Assumption 2') followed by Chernoff bound, we show that $T$ is a signature set with good probability. %By Chernoff bounds and using the upper bound on intersection size between two features, the lemma follows (details in Appendix~\ref{sec:lem:exists_sign}).%We uniformly randomly pick a set of size $t$ within these neighbors.  By slightly abuse of notation, we use $T$ to denote this random set in this proof. (Also note that we cannot run this probabilistic method in reality, because we don't know which are these neighbors). 
%For simplicity, 
%let $b = \sigma^2/(10\log n)$. 
For $k \ne j$, let $f_{k,T}$ be the number of edges from $x_k$ to $T$ in graph $G_{\tau}$. Then we can upperbound $\beta_{k,T}$ by $t\tau + f_{k,T}\Lambda$ since all edge weights are at most $\Lambda$ and there are at most $f_{j,T}$ edges with weights larger than $\tau$. Using simple Chernoff bound and union bound, we know that with probability at least $1-1/n$, for all $k\neq j$, $f_{k,T} \le 4\log n$. Therefore $\beta_{k,T}\le t\tau + f_{k,T}\Lambda \le \sigma^2t/(\Delta\log n)$ for $t \ge \Omega(\Lambda\Delta \log^2 n/\sigma^2)$, and $\tau = O(\sfrac{\sigma^2}{\Delta\log n})$.% and $\kappa \ge \Omega(\Lambda\log n)$, and $\tau \ge 2\Delta\log n$
\end{proof}

Although signature sets exist for all $x_j$, it is difficult to find them; even if we enumerate all subsets of size $t$, it is not clear how to know when we found a signature set.
 % (even if we tried all $t$-subsets).  
Thus we first look for ``correlated'' sets, which are defined as follows:

\begin{Def}[Correlated Set]\label{def:correlated_set}
%Take $p = ???$ sample vectors $y$, 
A set $T$ of size $t$ is called {\em correlated}, if with probability at least $\rho - 1/n^2$ over the choice of $x$'s, $\beta_T \ge 
%\Exp[S_T] 
\Exp[\beta_T] + 0.9\sigma t = t + 0.9\sigma t$. 
%(where $\E[S_T]$ is the empirical expectation of $S_T$).
\end{Def}

It follows easily (Lemma~\ref{lem:symmetric_diff}) that signature sets must be correlated sets.
\begin{Cor}\label{lem:sign-correlated}
If $T$ of size $t$ is a signature set for $x_j$, and $t = \omega(\sqrt{\log n})$, then $T$ is a correlated set. 
\end{Cor}

%The following technical lemma proves a stronger result than Lemma~\ref{lem:sign-correlated}. It illustrates that if $T$ is a signature set for variable $x_j$, then the event that $\beta_T$ is elevated is almost equivalent to the event that $x_j$ is activated. Thus assuming we have the a signature set $T$, not only is it a correlated set, but also we could use the correlation between coordinates in $T$ to decode the hidden value $x_j$ for most of the samples. 

%For any set $T$, we can rewrite $S_T$ as 

%Define $s_{j,T} = \sum_{i\in T} \mat{A}{i}{j}$ as the coefficient in front of $x_j$, we have 
%\begin{equation}
%S_T = \sum_{i\in T} y_i = \sum_j s_{j,T}x_j \label{eqn:1}
%\end{equation}

Although signature sets are all correlated sets, the other direction is far from true.
There can be many correlated sets that are not signature sets .  A simple counterexample would be that there are $j$ and $j'$ such that both $\beta_{j,T}$ and $\beta_{j',T}$ are larger than $\sigma t$. This kind of counterexample seems inevitable for any test on a set $T$ of polylogarithmic size.

%\Tnote{Should be Assumption ... instead of P1 P2}
% If we simply test signature set using correlation set, tn turns out that we will have huge false positives. A simple counterexample would be that there are $j$ and $j'$ such that both $\beta_{j,T}$ and $\beta_{j',T}$ are larger than $\sigma t$. This kind of counterexample seems inevitable for any test on a set $T$ of small size $t$, which is at most $poly log (n) $. 

To resolve this issue, the idea is 
%use Lemma~\ref{lem:symmetric_diff} 
to {\em expand} any set of size $t$ into a much larger set $\tilde{T}$, which is called {\em expanded set for $T$}. If $T$ happened to be a signature set to start with, such an expansion would give good estimate of the corresponding column of $A$, and more importantly, $\tilde{T}$ will have a similar `signature' property as $T$, which we can now verify because $\tilde{T}$ is large. 

Algorithm~\ref{alg:expand} and Definition~\ref{def:expand} show how to expand $T$ to $\tilde{T}$. The empirical expectation $\EExp[f(y)]$ is defined to be $\frac{1}{N}\sum_{i=1}^N f(y^i)$.

%Although we hope for all correlated set, there is a feature $j$ that connects to all of them in $G_\sigma$, this is not always true. However, there exists special correlated sets for each feature $x_j$, which we call signature sets

%If a set $T$ is highly correlated with one particular $x_j$, then we would expect all the samples with large $S_T$ to have $x_j = 1$. Using this idea we can try to estimate the column of $A$ via the following set expanding procedure

\begin{algorithm}
\caption{$\tilde{T}$ = expand($T$, threshold)}\label{alg:expand}
\begin{algorithmic}[1]
\REQUIRE $T$ of size $t$, $d$, and $N$ samples $\set{y^1,\dots,y^N}$
\ENSURE vector $\tilde{A}_{T}\in \R^n$ and expanded set $\tilde{T}$ of size $d$ 
(when $T$ is a signature set $\tilde{A}_T$ is an estimation of $A_j$).
\STATE {\em Recovery Step: } Let $L$ be the set of samples whose $\beta_T$ values are larger than $\EExp[\beta_T] + threshold$
 \[ L = \set{y^k \mid \beta_T^k \ge \EExp[\beta_T] + threshold}\] %\Tnote{replaced $\EExp[S_T]$ simply by $t$}
\STATE {\em Estimation Step: } Compute the empirical mean of samples in $L$, and obtain $\tilde{A}_T$ by shifting and scaling
\[\EExp_{L}[y] = \frac{1}{|L|}\left(\sum_{y^k \in  L} y^k\right), \textrm{ and }\tilde{A}_{T}(i) = \max\{0, (\EExp_{L}[y_i] - \EExp[y_i])/(1-\rho)\}\]
\STATE {\em Expansion Step: } $\tilde{T} = \{\textrm{$d$ largest coordinates of $\tilde{A}_T$}\}$
%\Tnote{Was $\EExp[y]$ instead of $\mathbf{1}$ before}
\end{algorithmic}
\end{algorithm}

\begin{Def}\label{def:expand}
For any set $T$ of size $t$, the expanded set $\tilde{T}$ for $T$ is defined as the one output by Algorithm~\ref{alg:expand}. The estimation $\tilde{A}_T$ is the output at step 2.
%let  $ L = \set{y^k | \beta_T^k \ge \EExp[\beta_T] + 0.9\sigma t}$, and let $\EExp_{L}[y]$ be the empirical mean of samples in $L$, that is, $\EExp_{L}[y] = \frac{1}{|L|}\left(\sum_{y^k \in  L} y^k\right)$. Define $\tilde{A}_{T}(i) = \max\{0,\frac{1}{1-\rho}(\EExp_{L}[y_i] - \EExp[y_i])\}$. The expanded set $\tilde{T}$ for $T$ is defined as $\tilde{T} = \{\textrm{$d$ largest coordinates of $\hat{A}_T$}\}$
\end{Def}

%\begin{Def} 
%For any set $T$, let $S$ be the set of samples with$Y_T \ge \E[Y_T + 0.9\sigma t]$, and $\E_S[y_i] = \fracp{\sum_{y\in S} y_i}{|S|}$. Let $V_{T,i} = \E_S[y_i] - \E[y_i]$, then the expanded set $D$ corresponding to $T$ contains $i\in[n]$ for $d$ largest $V_{T,i}$'s.
%\end{Def}

When $T$ is a signature set for $x_j$, then $\tilde{A}_T$ is already close to the true $A_j$, and the expanded set $\tilde{T}$ is close to the largest entries of $A_j$. 
\begin{Lem}\label{lem:extend}
If $T$ is a signature set for $x_j$ and the number of samples $N = \Omega(n^{2+\delta}/\rho^3)$, where $\delta$ is any positive constant, then with high probability $||\tilde{A}_T - A_j||_{\infty} \le 1/n$. 
Furthermore, $\beta_{j,\tilde{T}}$ is $d/n$-close to the sum of $d$ largest entries in $A_j$, and for all $i\in \tilde{T}$, $A_j^{(i)} \ge 0.9\sigma$.
\end{Lem}
%\Tnote{Should check every call of this lemma to make sure it doesn't miss something we had and the consistency of the notation. }

\begin{proof}
Let's first consider $\E[\tilde{A}_T] \triangleq (\Exp[y|E_2] - \mathbf{1})/(1-\rho)$ where $E_2$ is the event that $\beta_T \ge t + 0.9\sigma t$ defined in Lemma~\ref{lem:symmetric_diff}. %defined in Lemma~\ref{lem:symmetric_diff}, $E_2$ is defined as the event that $S_T > t + 0.9\sigma t$. 
Recall that because of normalization, we know for any $j$, $\sum_{i\in [n]}\mat{A}{i}{j} = 1/\rho$, so in particular $y_i \le 1/\rho$.
By Lemma~\ref{lem:symmetric_diff} and some calculations (see Lemma~\ref{lem:prob_ineq}), we have that $|\Exp[y|E_2] -\Exp[y|E_1]|_{\infty} \le n^{-C}/\rho$. Note that $\Exp[y|E_1] =  1 + (1-\rho)A_j$. Therefore we have that $|\E[\tilde{A}_T] - A_j|_{\infty} \le n^{-C}/\rho$.

Now by concentration inequalities when $N = \Omega(n^{2+\delta}/\rho^3)$ (notice that the variance of each coordinate is bounded by $\Lambda$), $\|\tilde{A}_T - \E[\tilde{A}_T]\|_\infty \le 1/n$ with very high probability ($\exp(-\Omega(n^\delta))$). This probability is high enough so we can apply union bound for all signature sets.
%Though we don't have access to $\tilde{A}_T$, we know that $\tilde{A}_T$ is efficient close to $\hat{A}_T$, since the later is simply the empirical estimate of the former. 
\end{proof}

%\Rnote{We probably don't need such a high accuracy}

%
%\begin{Lem}\label{lem:extend_old}
%If $T$ is a signature set for $x_j$, then difference $V_{T,i}$ is $???$ (something smaller than $\sigma$) close to $\mat{A}{i}{j}$. In particular, $s_{i,D}$ is within $???*d$ close to the sum of $d$ largest entries in $A_j$, and for all elements $u\in D$ $\mat{A}{u}{i} \ge 0.9\sigma$.
%\end{Lem}

%\begin{proof}
%
%First show the symmetric difference between $x_j = 1$ and $Y_T \ge \E[Y_T]+0.9\sigma t$ is small:
%$\Pr[Y_T \ge \E[Y_T]+0.9 \sigma t|x_j = 1] \ge 1 - ???$, $\Pr[x_j = 1 | Y_T \ge \E[Y_T] + 0.9 \sigma t ] \ge 1 - ???$.
%
%The first conditional probability is just applying concentration bounds on small $s_{j,T}$'s (recall small means smaller than $\sigma^2 t/\log n$).
%
%The second conditional probability uses the fact that if $Y_T$ is large and $x_j = 0$, then either the contribution from small $s_{j,T}$'s are far from mean, or there are at least two large $s_{j,T}$'s being on. Both events happen with small probability.
%
%Now using the bound on symmetric difference the range (or expectation for large entries) should give us the result.
%\end{proof}
\subsection{Identify Expanded Signature Sets}
\label{subsec:expandedsignatureset}
We will now see the advantage that the expanded sets $\tilde{T}$ provide. If $T$ happens to be a signature set, the expanded set $\tilde{T}$ for $T$ also has similar property.  But now $\tilde{T}$ is a much larger set (size $d$ as opposed to $t =polylog$), and we know (by Assumption 2') that different features have limited intersection, so if we see a large elevation it is likely to be caused by a single feature! We will leverage this in order to identify expanded signature sets among all the expanded sets.

%The advantage of expanding a signature set $T$ to an expanded signature set $\tilde{T}$ is because it is much easier identify expanded signature sets.
%\Tnote{might note that the gap is constant here opposing to the log gap in the def of signature set}
If an expanded set $\tilde{T}$ also has essentially a unique large coefficient $\beta_{j,\tilde{T}}$, we call it an {\em expanded signature set}.

\begin{Def}[Expanded Signature Set]
An expanded set $\tilde{T}$ is an expanded signature set for $x_j$ if $\beta_{j,\tilde{T}} \ge 0.7\sigma d$ and for all $k\ne j$, $\beta_{k,\tilde{T}} \le 0.3\sigma d$.
\end{Def}

Note that an expanded signature set always has size $d$ and the gap between largest $\beta_{j,T}$ and the second largest is only constant as opposed to logarithmic in the definition of signature set. As its name suggests, a expanded set $\tilde{T}$ of a signature set $T$ for $x_j$ is an expanded signature set for $x_j$ as well.  On one hand, the Lemma~\ref{lem:extend} guarantees that $\tilde{T}$ connects to $x_j$ with large weights, and on the other hand, since the pairwise intersection of neighborhoods of $x_j$ and $x_k$ in $G_{\tau}$ is small, $\tilde{T}$ cannot also connect to other $x_k$ with too many large weights. 

%This is OK since we will show that the number of these other large coefficients will be small.

%\Anote{almost like a signature set of size $d$ -- can we use old def?}
%\Rnote{No. We cannot identify the sets that satisfy the old def}
%\Tnote{Rong, this is why I thought expanded signature set is weaker than the signature set, since here only constant gap between large and small is needed, and expanded sign is easier to identify in some sense. Maybe we could make this point informally somewhere?}

\begin{Lem}
\label{lem:expandsignature}
If $T$ is a signature set for $x_j$, then the expanded set $\tilde{T}$ for $T$ is always an expanded signature set for $x_j$. In fact, the coefficient $\beta_{j,\tilde{T}}$ is at least $0.9\sigma d$.
\end{Lem}

\begin{proof}
Since we know there are at least $d$ weights $\mat{A}{i}{j}$ bigger than $\sigma$ for any column $A_j$, by Lemma~\ref{lem:extend} we know $\beta_{j,\tilde{T}} \ge \sigma d - o(1)d \ge 0.9\sigma d$.

Furthermore, Lemma~\ref{lem:extend} says $x_j$ connects to every node in $\tilde{T}$ with weights larger than $0.9\sigma$ (since by Assumption 1 there are more than $d$ edges of weight at least $\sigma$ from node $j$). By Assumption 2 on the graph, for any other $k\ne j$, the number of $y_i$'s that are connected to both $k$ and $j$ in $G_{\tau}$ is bounded by $\kappa$. In particular, the number of edges from $k$ to $\tilde{T}$ with weights more than $\tau$ is bounded by $\kappa$. Therefore the coefficient $\beta_{k,\tilde{T}} = \sum_{(i,k)\in G_{\tau}}\mat{A}{i}{k} + \sum_{(i,k)\not \in G_{\tau}}\mat{A}{i}{k} $ is bounded by  
$\Lambda\kappa + |\tilde{T}|\tau = o(d) \le 0.3d$. (Recall $\tau = o(1)$ and $\kappa = o(d)$)
%$\tau$ and $\kappa$ are $\omega(1)$). 
%$\sigma^2/20C^2 \log n \cdot |D| + \Lambda\cdot \sigma^2d/5C^3\log n \le 0.3\sigma d$.
\end{proof}

%\Tnote{$\tau$ was $20C^2\log n$ and $\kappa$ was $5C^3\log n$}

%Furthermore, Lemma~\ref{lem:extend} says $x_j$ connects to every node in $\tilde{T}$ with weights larger than $0.9\sigma$. By Assumption 2 on the graph, for any other feature  $k\ne j$, the number of pixels that are connected to both $k$ and $j$ in $G_{\sigma^2/\tau}$ is bounded by $\sigma^2d/5C^3\log n$. In particular, the number of edges from $i$ to $D$ with weight more than $\sigma^2/20C^2 \log n$ is bounded by $\sigma^2d/5C^3\log n$. The coefficient $s_{i,D}$ is at most $\sigma^2/20C^2 \log n \cdot |D| + \Lambda\cdot \sigma^2d/5C^3\log n \le 0.3\sigma d$.
%\end{proof}

The following notion of empirical bias is a more precise way (compared to correlated set) to measure the simultaneous elevation effect.% within the coordinates $y_i$ in a set $\tilde{T}$ of size $d$.  

%The advantage of expanding a set $T$ to an expanded signature set $\tilde{T}$ is that we could rely on the notation of empirical bias defined below to identify the expanded signature set $\tilde{T}$. Recall that $\beta^k_{\tilde{T}} = \sum_{i\in \tilde{T}} y^k_i$ is the value of $\beta_{\tilde{T}}$ for sample $y^k$. 

%In order to do this we rely on the following notion of empirical bias:

%The advantage of expanding a signature set $T$ to an expanded signature set $\tilde{T}$ is because it is much easier to identify expanded signature sets. In order to do this we rely on the following notion of empirical bias:

%We now know that if $\tilde{T}$ is an expanded set for signature set $T$, then $\tilde{T}$ is itself a expanded signature set of larger size. To identify an expanded signature set of size $d$ turns out to be much easier. To do this we define the empirical bias of a set $D$ of size $d$. 

%\Tnote{It would be nice if we can define the bias purely depending on $D$, but not on the samples, like in the definition of signature set and correlated set, but the empirical one seems easier to work on}

%In this step we would like to find $\tilde{T}$ that are expanded signature sets for all $x_j$'s. We know there must be such expanded signature sets from previous lemma. In order to find them we define the bias of a set:

\begin{Def}[Empirical Bias] The empirical bias $\bias_{\tilde{T}}$ of an expanded set $\tilde{T}$ of size $d$ is defined to be the largest $B$ that satisfies
$$\card{\set{k\in [p]:\beta^k_{\tilde{T}} \ge \EExp[\beta_{\tilde{T}}]+B}} \ge \rho N/2.$$

In other words, $\bias_{\tilde{T}}$ is the difference between the $\rho N/2$-th largest $\beta^k_{\tilde{T}}$ in the samples and $\EExp[\beta_{\tilde{T}}]$.
\end{Def}

%
%\begin{Def}[Bias] The bias of a set $D$ is the difference between the $\rho N/2$ largest $S_D$ and $\E[S_D]$.
%\end{Def}

The key lemma in this part shows the expanded set with largest empirical bias must be an expanded signature set:

\begin{Lem}\label{lem:maxbias}
Let $\tilde{T}^*$ be the set with largest empirical bias $\bias_{\tilde{T}^*}$ among all the expanded sets $\tilde{T}$. The set $\tilde{T}^*$ is an expanded signature set for some $x_j$.
\end{Lem}

We build this lemma in several steps. First of all, we show that the bias of $\tilde{T}$ is almost equal%captures how large could $\beta_{\tilde{T}}$ be with probability $\rho$, which is almost equivalent 
to the largest $\beta_{j,\tilde{T}}$:% for the following reason: 
if $\beta_{\tilde{T}}$ contains a large term $\beta_{j,\tilde{T}}x_j$, then certainly this term will contribute to the bias $\bias_{\tilde{T}}$; on the other hand, suppose in some extreme case $\beta_{\tilde{T}}$ only has two non-zero terms $\beta_{j,\tilde{T}}x_j + \beta_{k,\tilde{T}}x_k$. Then they cannot contribute more than $\max\{\beta_{j,\tilde{T}}, \beta_{k,\tilde{T}}\}$ to the bias, because otherwise both $x_k$ and $x_j$ have to be 1 to make the sum larger than $\max\{\beta_{j,\tilde{T}}, \beta_{k,\tilde{T}}\}$, and this only happens with small probability $\rho^2 \ll \rho$. 

The intuitive argument above %assumes unrealistically that $\beta_{\tilde{T}}= \beta_{j,\tilde{T}}x_j + \beta_{k,\tilde{T}}x_k$. However, this 
is not far from true: basically we could show that a) There are indeed very few large coefficients $\beta_{k,\tilde{T}}$'s  (see Claim~\ref{claim:largeedges} for the precise statement) b) the sum of  those small $\beta_{k,\tilde{T}}x_k$ concentrates around its mean, thus won't contribute much to the bias.  

After relating the bias of $\tilde{T}$ to the largest coefficients $\max_j \beta_{j,\tilde{T}}$, we further argue that taking the set $\tilde{T}^*$ with largest bias among all the $\tilde{T}$, we not only see a large coefficient $\beta_{j,\tilde{T}}$, but also we observe a gap between the the top $\beta_{j,\tilde{T}}$ and all other $\beta_{k,T}$'s, and hence $\tilde{T}$ is an expanded signature set for $x_j$. 

%\Tnote{The following para needs revise.. }
%The intuition behind the lemma is the following: if the empirical bias for a set $\tilde{T}$ is larger than $0.8 d \sigma$, then we show that there must be at least one $j$ with $\beta_{j, \tilde{T}} > 0.7 d \sigma$.  This is by contradiction: if not, then in a $\rho/2$ fraction of the samples, we must have had a contribution from two distinct $j$'s, which means there are {\em many} $j$'s with `large' $\beta_{j, \tilde{T}}$ (since any two $j$'s are both {\em on} with probability only $\rho^2$).  But this turns out to give a contradiction to the intersection bound!  So we now need to rule out the possibility that two distinct $j$'s (say $j$ and $j'$) have $\beta_{j, \tilde{T}} > 0.7 d\sigma$. In this case, there must be many edges from both $j$ and $j'$ to $\tilde{T}$. But they cannot intersect too much by assumption.  Thus the set $\tilde{T}$ must have ``many'' vertices that are not neighbors of $j$.  We can then use this to contradict the maximum bias assumption of $\tilde{T}$ (in particular, the extension of a signature set of $j$ typically has a higher bias), completing the proof.

%efore proving this lemma, let us first look at the possible coefficients $\beta_{j,\tilde{T}}$ for an expanded set $\tilde{T}$. 

We make the arguments above precise by the following claims. First, we shall show there cannot be too many large coefficients $\beta_{j,D}$ for any set $D$ of size $d$ (although we only apply the claim on expanded sets).

\begin{claim}
\label{claim:largeedges}
For any set $\tilde{T}$ of size $d$, the number of $k$'s such that $\beta_{k,\tilde{T}}$'s is larger than 
$\sfrac{d\sigma^4}{\Delta \Lambda^2 \log n}$
 is at most 
$O(\Delta\Lambda^3\log n/\sigma^4)$. 
%Here $\Delta$ is a large enough constant.
\end{claim}
\newcommand{\KL}{K_{large}}

%-------------------------THIS IS WHERE \tau GETS SET-----------------------%
%
%
%\Tnote{Need $\tau \le \sfrac{\sigma^4}{2\Delta\Lambda^2\log n}$}
%
%
%-------------------------THIS IS WHERE \tau GETS SET-----------------------%

\begin{proof}
%Suppose for some feature $k$, we know
For the ease of exposition, we define $\KL = \{k: \beta_{k,\tilT}\ge \sfrac{d\sigma^4}{\Delta\Lambda ^2\log n}\}$.  
%Let $K$ be the set of $k$ such that $\beta_{k,\tilde{T}}$ is more than $\sfrac{d\sigma^2}{\Delta \Lambda^2 \log n}$. 
Hence the goal is to prove that $|\KL| \le O(\Delta\Lambda^3\log n/\sigma^4)$. 
Recall that $\beta_{k,\tilde{T}} = \sum_{i\in \tilde{T}} \mat{A}{i}{k}$. Let $Q_k = \{i \in \tilde{T} : \mat{A}{i}{k} \ge \tau\}$ be the subset of nodes in $\tilT$ that connect to $k$ with weights larger than $\tau$. We have that $\beta_{k,\tilde{T}} = \sum_{i\not \in Q_k}\mat{A}{i}{k} + \sum_{i \in Q_k}\mat{A}{i}{k}$. The first sum is upper bounded by $d\tau \le \sfrac{d\sigma^4}{2\Delta\Lambda^2\log n}$. Therefore for $k \in \KL$, the second sum is lower bounded by $\sfrac{d\sigma^4}{2\Delta\Lambda^2\log n}$. Since $\mat{A}{i}{k}\le \Lambda$, we have $|Q_k|\ge \sfrac{\sigma^4d}{2\Delta\Lambda^3\log n}$.

%The contribution to $\beta_{k,\tilde{T}}$ of weights $\mat{A}{i}{k}$ that are smaller than $\tau$ is upperbounded by $d\tau$, which is less than $\sfrac{d\sigma^4}{2\Delta\Lambda^2\log n}$, since $\tau = O(\sfrac{d\sigma^4}{\Delta\Lambda^2\log n})$. Thus those weights $\mat{A}{i}{k}$ that are larger than $\tau$ account for the remaining $\sfrac{d\sigma^4}{2\Delta\Lambda^2\log n}$ mass of $\beta_{k,\tilT}$.

%Therefore the remaining $\sfrac{\sigma^2}{2\Delta\Lambda^2\log n}$ contribution of $\beta_{k,\tilde{T}}$ is from edges in $G_{\tau}$. 
%Since the weights have upper bound $\Lambda$, there are least $\sfrac{\sigma^2d}{2\Delta\Lambda^3\log n}$ such edges connected to $\tilde{T}$ for node $k$. In other words, $|Q_k|\ge \sfrac{\sigma^2d}{2\Delta\Lambda^3\log n}$.
%terms with large weights. These terms all correspond to edges in $G_{\sigma^2/\tau}$.

On the other hand, by Assumption 2 we know in graph $G_{\tau}$, any two features cannot share too many pixels:
%For any feature $j$ with $\beta_{j,\tilde{T}}$ larger than $\sfrac{d\sigma^2}{2\Delta\Lambda^2 \log n}$, let $Q_j\subset \tilde{T}$ be the subset of pixels that are connected to $j$ in $G_{\tau}$. %We know $\card{Q_j} \ge \sigma^2d/2\Delta\Lambda^3\log n$, 
for any $k$ and $k'$, $\card{Q_k\cap Q_{k'}} \le \kappa$. Also note that by definition, $Q_j\subset \tilde{T}$, which implies that $|\cup_{k\in \KL}Q_k| \le |\tilde{T}| = d$. 
%$\card{\cup_{j:s_{j,\tilde{T}}\ge d\sigma^2/\Delta\Lambda^2 \log n}  Q_j} \le d$. However, when the number of such sets is larger than $k = 4\Delta\Lambda^3\log n/\sigma^2$, 
By inclusion-exclusion we have
\begin{equation}
%\card{\cup_{j:s_{j,\tilde{T}}\ge d\sigma^2/\Delta\Lambda^2\log n}  Q_j} \ge k \cdot \frac{\sigma^2d}{2\Lambda^3\Delta\log n} - {k\choose 2}\cdot \kappa > d.
d \ge |\bigcup_{k\in \KL}Q_k| \ge \sum_{k\in \KL} |Q_k| - \sum_{k,k'\in \KL}\card{Q_k\cap Q_{k'}}\ge |\KL|\sfrac{\sigma^4d}{2\Delta\Lambda^3\log n} - |\KL|^2/2 \cdot\kappa \label{eqn:`in-ex}
\end{equation}
This implies that $|\KL|\le  O(\Delta\Lambda^3\log n/\sigma^4)$, when $\kappa = O(\sfrac{\sigma^8d}{ \Delta^2\Lambda^6\log^2 n})$. \footnote{Note that any subset of $\KL$ also satisfies equation~(\ref{eqn:`in-ex}), thus we don't have to worry about the other range of the solution of (\ref{eqn:`in-ex})}
%This contradicts with the fact that all $Q_j$'s are in $\tilde{T}$. Hence $k$ is smaller than $O(\Delta\Lambda^3\log n/\sigma^2)$.
\end{proof}

%See Appendix~\ref{sec:claim:largeedges} for a proof.
%\Rnote{$\kappa$ and $\tau$ should be chosen according to the lemma above}

%We call coefficients that are larger than $\sfrac{d\sigma^2}{\Delta\Lambda^2\log n}$ large, and the rest of coefficients are small.

For simplicity, let $k^* = \arg\max_k \beta_{k,\tilde{T}}$, so $\beta_{k^*,\tilde{T}}$ is the largest coefficient in $\beta_{\tilde{T}}$. Recall that the definition of expanded signature set roughly translates to a constant factor gap between $\beta_{k^*,\tilde{T}}$ to any other coefficient $\beta_{k,\tilde{T}}$. 

The next claim shows that the empirical bias $\bias_{\tilde{T}}$ is a good estimate of $\beta_{k^*, \tilde{T}}$ when $\beta_{k^*, \tilde{T}}$ is large. 
%This claim in particular shows the bias of a set is very close to its largest $s_{j,\tilde{T}}$ value.

\begin{claim}
\label{claim:biasmaxclose}
For any expanded $\tilde{T}$ of size $d$, with high probability over the choices of all the $N$ samples, the empirical bias $\bias_{\tilde{T}}$ is within $0.1d\sigma^2/\Lambda$ to $\beta_{k^*,\tilde{T}} = \max_{k} \beta_{k,\tilde{T}}$ when $\beta_{k^*,\tilde{T}}$ is at least $0.5d\sigma$.
\end{claim}

\begin{proof}
Let $\KL' = \KL\setminus\set{k^*}$ \footnote{$\KL$ is defined in proof of Claim~\ref{claim:largeedges}}, %be set of $k$ such that $\beta_{k,\tilde{T}}$ is more than $\sfrac{d\sigma^2}{\Delta \Lambda^2 \log n}$ except $k^*$. 
and $\beta_{small, \tilde{T}} = \sum_{k\not\in \KL} \beta_{k,\tilde{T}} x_k$, and $\beta_{large, \tilde{T}} = \sum_{k\in \KL'} \beta_{k,\tilde{T}} x_k$. 

First of all, the variance of $\beta_{small,\tilT}$ is bounded by $\rho\sum_{k\not\in \KL}\beta^2_{k,\tilT} \le \sfrac{d\sigma^4}{\Delta\Lambda^2\log n}\cdot \left(\rho\sum_{k\not\in \KL}\beta_{k,\tilT}\right) \le \sfrac{d^2\sigma^4}{\Delta\Lambda^2\log n}$. By Bernstein's inequality, for sufficiently large $\Delta$, with probability at most $1/n^2$ over the choice of $x$, the value $|\beta_{small, \tilde{T}} - \E[\beta_{small, \tilde{T}}]|$ is larger than $0.05d\sigma^2/\Lambda$, that is, $\beta_{small, \tilde{T}}$ nicely concentrates around its mean. Secondly, with probability at most $\rho$ we have $x_{k^*} = 1$ ,  and then $\beta_{k^*, \tilde{T}}x_{k^*}$ is elevated above its mean by roughly $\beta_{k^*, \tilde{T}}$.  Thirdly, the mean of $\beta_{large, \tilde{T}}$ is at most $\rho \sum_{k\in \KL'} \beta_{k,\tilde{T}}  \le \rho |K|d$, which is $o(\sigma d)$ by Claim~\ref{claim:largeedges}.  These three points altogether imply that with probability at least $\rho - n^{-2}$, $\beta_{\tilde{T}}$ is above its mean by $\beta_{k^*, \tilde{T}} - 0.1\sigma^2 d/\Lambda$. Also note that the empirical mean $\EExp[\beta_{\tilde{T}}]$ is sufficiently close to the $\beta_{\tilde{T}}$ with probability $1-\exp(-\Omega(n))$ over the choices of $N$ samples, when $N = poly(n)$. Therefore with probability $1-\exp(-\Omega(n))$ over the choices of $N$ samples, $\bias_{\tilde{T}} > \beta_{k^*, \tilde{T}} - 0.1\sigma^2 d/\Lambda$.  

It remains to prove the other side of the inequality, that is, $\bias_{\tilde{T}} \le \beta_{k^*, \tilde{T}} + 0.1\sigma^2 d/\Lambda$. 

%For the sake of contradiction we assume that for every $k$, $\beta_{k, \tilde{T}} < \bias_{\tilde{T}} - 0.1\sigma d/\Lambda$. 
Note that $|\KL|=O(\log n)$, thus with probability at least $1-2\rho^2|K|^2$, at most one of the $x_k,(k\in \KL)$ is equal to 1.   
%at most one of the $x_k, k\in K$ is 1. 
Then with probability at least $1-2\rho^2|K|^2$ over the choices of $x$, $\beta_{large, \tilde{T}}+ \beta_{k^*,\tilde{T}}$ is elevated above its mean by at most $\beta_{k^*, \tilde{T}}$. Also with probability $1-n^{-2}$ over the choices of $x$, $\beta_{small, \tilde{T}}$ is above its mean by at most $0.1\sigma^2 d/\Lambda$. Therefore with probability at least $1-3\rho^2|K|^2$ over the choices of $x$, $\beta_{\tilde{T}}$ is above its mean by at most $\beta_{k^*, \tilde{T}} + 0.1\sigma d/\Lambda$. Hence when $3\rho^2|K|^2 \le \rho/3$, with probability at least $1 - \exp(-\Omega(n))$ over the choice of the $N$ samples, $\bias_{\tilde{T}} \le \beta_{k^*, \tilde{T}} + 0.1\sigma^2 d/\Lambda$. The condition is satisfied when $\rho \le c/\log^2 n$ for a small enough constant $c$. 
\end{proof}

%See Appendix~\ref{sec:claim:biasmaxclose} for the proof.
Now we are ready to prove Lemma~\ref{lem:maxbias}.

\begin{proof}[of Lemma~\ref{lem:maxbias}]
By Claim~\ref{claim:biasmaxclose} and the existence of good expanded signature sets (Lemma~\ref{lem:expandsignature}), we know the maximum bias is at least $0.8\sigma d$. Apply Claim~\ref{claim:biasmaxclose} again, we know for the set $\tilde{T}^*$ that has largest bias, there must be a feature $j$ with $\beta_{j,\tilde{T}^*} \ge 0.7\sigma d$.

For the sake of contradiction, now we assume that the set $\tilde{T}^*$ with largest bias is not an expanded signature set. Then there must be some $k\ne j$ where $\beta_{k,\tilde{T}^*} \ge 0.3\sigma d$. Let $Q_j$ and $Q_{k}$ be the set of nodes in $\tilde{T}^*$ that are connected to $j$ and $k$ in $G_{\tau}$ (these are the same $Q$'s as in the proof of Claim~\ref{claim:largeedges}). We know $\card{Q_j\cap Q_{k}} \le \kappa$ by assumption, and $\card{Q_{k}} \ge 0.3\sigma d/\Lambda$. This means $\card{Q_j} \le d - 0.3\sigma d/\Lambda + \kappa$ by inclusion-exclusion.

Now let $T'$ be a signature set for $x_j$, and let $\tilde{T'}$ be its expanded set, from Lemma~\ref{lem:extend} we know $\beta_{j,\tilde{T'}}$ is almost equal to the sum of the $d$ largest entries in $A_j$, which is at least $0.2\sigma^2 d/\Lambda$ larger than $\beta_{j,\tilde{T}^*}$, since $|Q_j| \le d - 0.2\sigma d/\Lambda$. By Claim~\ref{claim:biasmaxclose} we know $\bias(\tilde{T'})\ge \beta_{j,\tilde{T'}}-0.1\sigma^2 d/\Lambda > \beta_{j,\tilde{T}^*}+0.1\sigma^2 d/\Lambda \ge \bias(\tilde{T})$, which contradict with the assumption that $\tilde{T}^*$ is the set with largest bias.
\end{proof}

%\Tnote{Stop editing here. }
Now we have found expanded signature sets, we can then apply Algorithm~\ref{alg:expand} (but with threshold $0.6\sigma d$ instead of $0.9\sigma d$) on that to get an estimation.

%\Rnote{the following lemma gives a bound on how large $1/\rho$ needs to be}

\begin{Lem}
\label{lem:expandexpandedset}
If $\tilde{T}$ is an expanded signature set for $x_j$, and $\tilde{A}_{\tilde{T}}$ is the corresponding column output by Algorithm~\ref{alg:expand}, then with high probability $\|\tilde{A}_{\tilde{T}} - A_j\|_\infty \le O(\rho(\Lambda^3\log n/\sigma^2)^2\sqrt{\Lambda\log n}) = o(\sigma)$.
\end{Lem}
\begin{proof}
Define $E_1$ to be the event that $x_j = 1$, and $E_2$ to be the event that $\beta_{\tilde{T}} \ge 0.6 d\sigma$.

When $E_1$ happens, event $E_2$ always happen unless $\beta_{\tilde{T},small}$ is far from its expectation. In the proof of Claim~\ref{claim:biasmaxclose} we've already shown the number of such samples is at most $n$ with very high probability.

Suppose $E_2$ happens, and $E_1$ does not happen. Then either $\beta_{\tilde{T},small}$ is far from its expectation, or at least two $x_j$'s with large coefficients $\beta{j,\tilde{T}}$'s are on. Recall by Claim~\ref{claim:largeedges} the number of $x_j$'s with large coefficients is $\card{K} \le O(\Lambda^3\log n/\sigma^2)$, so the probability that at least two large coefficient is ``on'' (with $x_j = 1$) is bounded by $O(\rho^2 \cdot \card{K}^2) = \rho \cdot O(\rho \Lambda^6\log ^2 n/\sigma^4) = \rho \cdot o(\sigma/\sqrt{\Lambda \log n}).$ With very high probability the number of such samples is bounded by $\rho N \cdot o(\sigma/\sqrt{\Lambda \log n})$.

Combining the two parts, we know the number of samples that is in $E_1\oplus E_2$ (the symmetric difference between $E_1$ and $E_2$) is bounded by $\rho N \cdot o(\sigma/\sqrt{\Lambda \log n})$. Also, with high probability $(1-n^{-C})$ all the samples have entries bounded by $O(\sqrt{\Lambda \log n})$ by Bernstein's inequality (variance of $y_i$ is bounded by $\sum_j \rho (\mat{A}{i}{j})^2 \le \max_j \mat{A}{i}{j} \sum_j \rho \mat{A}{i}{j}\le \Lambda$). Notice that this is a statement of the entire sample independent of the set $T$, so we do not need to apply union bound over all expanded signature sets.

Therefore by Lemma~\ref{lem:prob_ineq} $$\|\tilde{A}_{\tilde{T}} - A_j\|_\infty \le o(\sigma/\sqrt{\Lambda \log n})\cdot O(\sqrt{\Lambda \log n}) = o(\sigma).$$
\end{proof}

The previous lemma looks very similar to the lemma for signature sets, however, the benefit is we know how to find a set that is guaranteed to be expanded signature set! So we can iteratively find all expanded signature sets.

After identifying $\tilde{T}_1$, $\tilde{T}_2$, ..., $\tilde{T}_k$ (reorder the columns of $A$ to make them correspond to the first $k$ columns), we can estimate the corresponding columns $\tilde{A}_{\tilde{T}_1},\dots\tilde{A}_{\tilde{T}_k}$. Since these are close to the true columns $A_1, A_2, ..., A_k$ (wlog. we reorder columns so $\tilde{A}_{\tilde{T}_j}$ correspond to $A_j$ for $1\le j \le k$), we can in fact compute $\hat{\beta}_{j,\tilde{T}} = \sum_{i\in \tilde{T}} \tilde{A}_{\tilde{T}_j}(i)$. By Lemma~\ref{lem:expandexpandedset} we know $|\hat{\beta}_{j,\tilde{T}} - \beta_{j,\tilde{T}}| = o(\sigma d)$.

\begin{Lem}
Having found $\tilde{T}_i$ (and hence also $\tilde{A}_{\tilde{T}_i}$) for $i\le k$, let $\tilde{T}$ be the set with largest empirical bias among the expanded sets that have $\hat{\beta}_{j,\tilde{T}} < 0.2\sigma d$ for all $j \le k$. Then $\tilde{T}$ is an expanded signature set for new $x_j$ where $j > k$.  
\end{Lem}

\begin{proof}
The proof is almost identical to Lemma~\ref{lem:maxbias}. 

First, if $T$ is a signature set of $x_j$ where $j>k$, then by Lemma~\ref{lem:expandsignature} $\tilde{T}$ must satisfy $\hat{\beta}_{j,\tilde{T}} < 0.2\sigma d$, so it will compete for the set with largest empirical bias.

Also, since $\hat{\beta}_{j,\tilde{T}} < 0.2\sigma d$, we know the coefficients in $\beta_{j,\tilde{T}}$ must have $j > k$. Leveraging this observation in the proof of Lemma~\ref{lem:maxbias} gives the result.
\end{proof}

\subsection{Getting an Equivalent Dictionary}
\label{subsec:refine}
After finding expanded signature sets, we already have an estimation $\tilde{A}_{\tilde{T}_j}$ of $A_j$ that is entry-wise $o(\sigma)$ close. However, this alone does not imply that the two dictionaries are $\epsilon$-equivalent for very small $\epsilon$.

In the final step, we look at {\em all} the large entries in the column $A_j$, and use them to identify whether feature $x_j$ is 1 or 0. The ability to do this justifies the individually recoverable property of the dictionary.

\begin{Lem}\label{lem:refinement}
Let $S_j$ be the set of all entries larger than $\sigma/2$ in $\tilde{A}_{\tilde{T}_j}$, then $\card{S_j} \ge d$, $\beta_{j,S_j} \ge (0.5-o(1))\card{S_j}\sigma$, and for all $k\ne j$ $\beta_{k,S_j} \le \sigma^2\card{S_j}/\Delta \log n$ where $\Delta$ is a large enough constant.
\end{Lem}

\begin{proof}
This follows directly from the assumptions. By Assumption 1, there are at least $d$ entries in $A_j$ that are larger than $\sigma$, all these entries will be at least $(1-o(1))\sigma$ in $\tilde{A}_{\tilde{T}_j}$, so $\card{S_j} \ge d$.

Also, since for all $i\in S_j$, $\tilde{A}_{\tilde{T}_j}(i) \ge 0.5\sigma$, we know $A_j(i) \ge 0.5\sigma - o(\sigma)$, hence $\beta_{j,S_j}\ge (0.5-o(1))\card{S_j}\sigma$.

By Assumption 2, for any $k\ne j$, the number of edges in $G_{\tau}$ between $k$ and $S_j$ is bounded by $\kappa$, so $\beta_{k,S_j} \le \tau\card{S_j}+\kappa \Lambda \le \sigma^2\card{S_j}/\Delta \log n$.
\end{proof}

Since $S_j$ has a unique large coefficient $\beta_{j,S_j}$, and the rest of the coefficients are much smaller, when $\Delta$ is large enough, and $N \ge n^{4C+\delta}/\rho^3$ we know $\hat{A}_j$ is entry-wise $n^{-2C}/\log n$ close to $A_j$ (this is using the same argument as in Lemma~\ref{lem:extend}). We shall show this is enough to proof $n^{-C}$-equivalence between $\hat{A}$ and $A$.

%To start this entire description, we first note that the above process can ignore edges of small weight.
% $\epsilon$-equivalent dictionary $\hat{A}$, we first observe that it suffices to learn all the edges with relatively large weights:

\begin{Lem}
\label{lem:finalset}
Let $A, \hat{A}$ be dictionaries with rows having $\ell_1$-norm $O(1/\rho)$  
%$\norm{\cdot}_1 = 1/\rho$, 
and all entries in $A -
\hat{A}$ have magnitude at most $\delta$. Then $\hat{A}$ and $A$ are $O(\sqrt{\delta \log n})$-equivalent.
\end{Lem}
The proof is an easy application of Bernstein's inequality (see Appendix~\ref{sec:lem:finalset}).

\paragraph{Remark:} Notice that when $C \ge 1$ it is clear why $\hat{A}_j$ should have $\ell_1$ norm $1/\rho$ (because it is very close to $A_j$); when $C$ is smaller we need to truncate the entries of $\hat{A}_j$ that are smaller than $n^{-2C}/\log n$.

We now formally write down the steps in the algorithm.
%With this lemma we can restrict our attention to edges that have weights at least $\sigma$. Further, it is OK to round the edge weights to multiples of $\sigma$.
%From now on, $G_{\sigma}$ describes the bipartite graph consisting of the support of entries of $A$ that exceed $\sigma$.

%Let $G_{\delta}$ be a bipartite graph between $m$ and $n$ nodes. The edge $(i,j)$ is in the graph if $\mat{A}{i}{j}$ is at least $\delta$. The nodes on the $m$ side correspond to entries of $x$, and we %call them features; the nodes on the $n$ side correspond to entries of $y$, and we call them observed nodes. When we do not specify the graph, we work with the graph $G_{\sigma}$.
\begin{algorithm}
\caption{Nonnegative Dictionary Learning} \label{alg:main}
\begin{algorithmic}[1]
\REQUIRE $N$ samples $\set{y^1,\dots,y^N}$ generated by $y^i = Ax^i$. Unknown dictionary $A$ satisfies Assumptions 1 and 2.
\ENSURE  $\hat{A}$ that is $n^{-C}$ close to $A$
\STATE Enumerate all sets of size $t = O(\Lambda \log^2 n/\sigma^4)$, keep the sets that are correlated.
\STATE Expand all correlated sets $T$, $\tilde{T} = Expand(T, 0.9\sigma t)$.
\FOR{$j = 1$ TO $m$}
\STATE Let $\tilde{T}_j$ be the set with largest empirical bias, and for all $k < j$, $\hat{\beta}_{k,\tilde{T}} = \sum_{i\in T} \tilde{A}_{\tilde{T}_k}(i) \le 2d\sigma$.
\STATE Let $\tilde{A}_{\tilde{T}_k}$ be the result of estimation step in $Expand(\tilde{T}, 0.6\sigma d)$.
\ENDFOR
\FOR{$j = 1$ TO $m$}
\STATE Let $S_j$ be the set of entries that are larger than $\sigma/2$ in $\tilde{A}_{\tilde{T}_j}$
\STATE Let $\hat{A}_i$ be the result of estimation step in $Expand(S_j, 0.4\sigma \card{S_j})$
\ENDFOR
\end{algorithmic}
\end{algorithm}

\subsection{Working with Assumption 2}
\label{subsec:assumption2}
In order to assume Assumption 2 instead of 2', we need to change the definition of signature sets to allow $o(1/\sqrt{\rho})$ ``moderately large'' ($\sigma t/10$) entries. This makes the definition look similar to expanded signature sets. Such signature sets still exist by similar probabilistic argument as in Lemma~\ref{lem:exists_sign}. Lemma~\ref{lem:expandsignature} and Claims~\ref{claim:largeedges} and \ref{claim:biasmaxclose} can also be adapted.

Finally, for Lemma~\ref{lem:finalset}, the guarantee will be weaker (there can be $o(1/\sqrt{\rho})$ moderately large coefficients). The algorithm will only estimate $x_j$ incorrectly if at least $6$ such coefficients are ``on'' (has the corresponding $x_j$ being 1), which happens with less than $o(\rho^3)$ probability. By argument similar to Lemma~\ref{lem:extend} and Lemma~\ref{lem:finalset} we get the first part of Theorem~\ref{thm:nonneg:main}.
\section{General Case}\label{sec:general}

With minor modifications, our algorithm and its analysis can be adapted to the general case in which the matrix $A$ can have both positive and negative entries.

%The aim again is to obtain an entry wise approximation $\hat{A}$. 
We follow the outline from the non-negative case, and look at sets $T$ of size $t$.  The quantities $\beta_T$ and $\beta_{j,T}$ are defined exactly the same as in Section~\ref{subsec:signaturesets}.  Additionally, let $\nu_T$ be the standard deviation of $\beta_T$, and let $\nu_{-j,T}$ be the standard deviation of $\beta_T - \beta_{j,T} x_j$. That is,
\[\nu_{-j,T}^2 = \Var[\beta_T - \beta_{j,T} x_j] = \rho \sum_{k\neq j}^{} \beta^2_{k,T}.\]

The definition of signature sets requires an additional condition to take into account the standard deviations. 

%\Tnote{does this $\Var$ notation for variance sound good }
\begin{Def}[(General) Signature Set]\label{def:sign-general}
A set $T$ of size $t$ is  a {\em signature set} for $x_j$, if for some large constant $\Delta$, we have: (a) $|\beta_{j,T}| \ge \sigma t$, (b) for all $k\ne j$, the contribution $|\beta_{k,T}| \le \sigma^2 t/(\Delta\log n)$, and additionally, (c) $\nu_{-j,T} \le \sigma t/\sqrt{\Delta \log n}$.
\end{Def}

In the nonnegative case the additional condition $\nu_{-j,T} \le \sigma t/\sqrt{\Delta \log n}$ was automatically implied by nonnegativity and scaling.  Now we use Assumption G3 to show there exist $T$ in which (c) is true along with the other properties.  To do that, we prove a simple lemma which lets us bound the variance (the same lemma is also used in other places).

\begin{Lem}\label{lem:variance-bound-general}
Let $T$ be a set of size $t$ and $S$ be an arbitrary subset of features, and consider the sum $\beta_{S,T} = \sum_{j\in S} \beta_{j,T}x_j$. Suppose for each $j\in S$, the number of edges from $j$ to $T$ in graph $G_{\tau}$ is bounded by $W$. Then the variance of $\beta_{S,T}$ is bounded by $2tW + 2t^2\gamma$.
\end{Lem}
\begin{proof}
The idea is to split the weights $\mat{A}{i}{j}$ into the {\em big} and {\em small} ones (threshold being $\tau$). Intuitively, on one hand, the contribution to the variance from large weights is bounded above because the number of such large edges in bounded by $W$. On the other hand, by assumption (3), the total variance of small weights is less than $\gamma$, which implies that the contribution of small weight to the variance is also bounded. Formally, we have 
%For the rest of the proof, we will drop the subscript in $G$ for ease of notation.
\begin{eqnarray*}
\Var[\beta_{S,T}] = \rho \sum_{j\in S} \beta_{j,T}^2&=&  \rho \sum_{j\in S} \left(\sum_{i\in T} \mat{A}{i}{j}\right)^2 \\
&=&  \rho \sum_{j\in S} \left(\sum_{i: i\in T, (i,j)\in G_{\tau }} \mat{A}{i}{j} + \sum_{i: i\in T, (i,j) \not \in G_{\tau}} \mat{A}{i}{j}\right)^2 \\
&\le &  2\rho \sum_{j\in S}\left[ \left(\sum_{i: i\in T, (i,j)\in G_{\tau}} \mat{A}{i}{j}\right)^2+\left(\sum_{i: i\in T, (i,j)\not\in G_{\tau}} \mat{A}{i}{j}\right)^2 \right]\\
&\le&  2\rho \sum_{j\in S} \left[W\left(\sum_{i: i\in T, (i,j)\in G_{\tau}} \left(\mat{A}{i}{j}\right)^2\right)
%&\le&  2\rho \sum_{j\in S} \left(\sum_{i: (i,j)\in G} \One(i\in T)^2\right)\left(\sum_{i: (i,j)\in G} \left(\mat{A}{i}{j}\right)^2\right)+ \\
+ t\left(\sum_{i: i\in T, (i,j)\not\in G_{\tau}} \left(\mat{A}{i}{j}\right)^2\right)\right]\\
&=& 2\rho W \sum_{i\in T}\sum_{j\in S}\left(\mat{A}{i}{j}\right)^2 + 2\rho t\sum_{i\in T}\sum_{j: (i,j)\not\in G_{\tau }}\left(\mat{A}{i}{j}\right)^2\\
&\le& 2 t W + 2t^2 \gamma.
\end{eqnarray*}
In the fourth line we used Cauchy-Schwarz inequality and in the last step, we used Assumption G3 about the total variance due to small terms being small, as well as the normalization of the variance in each pixel.
\end{proof}

%This lemma is required in the proof of Lemma~\ref{lem:symmetric_diff} and Claim~\ref{claim:biasmaxclose}. Also, for Claim~\ref{claim:biasmaxclose}, the large coefficients should be think of as $\beta_{j,\tilde{T}}$'s with many edges between $j$ and $\tilde{T}$ in $G_{\sigma^2/\tau}$ (before this definition is very similar to $\beta_{j,\tilde{T}}$ is large, but when there are negative weights they can be very different.

\begin{Lem}\label{lem:exists_sign_neg}
Suppose $A$ satisfies our assumptions for general dictionaries, and let $t = \Omega(\Lambda\Delta\log^2 n/\sigma^2)$. Then for any $j \in [n]$, there exists a general signature set of size $t$ for node $x_j$ (as in Definition~\ref{def:sign-general}).
\end{Lem}%\Anote{check parameter $t$}
\begin{proof}
As before, we use the probabilistic method. Suppose we fix some $j$.  By Assumption G1, in $G_{\sigma}$, node $x_j$ has either at least $d$ positive neighbors or $d$ negative ones. W.l.o.g., let us assume there are $d$ negative neighbors. Let $T$ be uniformly random subset of size $t$ of these negative neighbors. By definition of $G_{\sigma}$, we have $\beta_{j,T} \le -\sigma t$. 

For $k \ne j$, let $f_{k,T}$ be the number of edges from $x_k$ to $T$ in graph $G_{\tau}$. Using the same argument as in the proof of Lemma~\ref{lem:exists_sign}, we have $f_{k,T}\le 4\log n$ w.h.p. for all such $k\neq j$. Thus $|\beta_{k,T}|\le t\tau + f_{k,T}\Lambda \le \sigma^2t/(\Delta \log n)$. Thus it remains to bound $\nu_{-j,T}$.

We could apply Lemma~\ref{lem:variance-bound-general} with $W = 4\log n\ge f_{k,T}$, and $S = [m]\setminus \{j\}$ on set $T$: we get $\nu_{-j,T}^2 \le 2tW + 2t^2\gamma 
%\le 8t\log n + \sfrac{2t^2\sigma^2}{3\Delta^2\log n}
$. Recall that $\gamma = \sfrac{\sigma^2}{3\Delta^2\log n}$ and thus $\nu_{-j,T} \le \sfrac{\sigma t}{\sqrt{\Delta\log n}}$.
 %to obtain a variance bound of $8t\log n + 2 t^2 \gamma^2$. Now observing that $\gamma^2$ is defined to be $\sigma^2/(\Delta^2 \log n)$ completes the proof.
\end{proof}
%----------------------ONE CONSTRAINT FOR \gamma-----------------
%
%
%\Tnote{$\gamma \le \sigma^2/3\Delta^2\log n$}
%
%
%
%----------------------ONE CONSTRAINT FOR \gamma-----------------

The proof of Lemma~\ref{lem:symmetric_diff} now follows in the general case (here we will use the variance bound (c) in the general definition of signature sets), except that we need to redefine event $E_2$ to handle the negative case. For completeness, we state the general version of Lemma~\ref{lem:symmetric_diff} in Appendix~\ref{sec:app:symmetric_diff}. As before, signature sets give a great idea of whether $x_j=1$.

Let us now define correlated sets: here we need to consider both positive and negative bias
 
\begin{Def}[(General) Correlated Set]\label{def:correlated_set_neg}
%Take $p = ???$ sample vectors $y$, 
A set $T$ of size $t$ is {\em correlated}, if either with probability at least $\rho - 1/n^2$ over the choice of $x$'s, 
%we have 
$\beta_T \ge 
%\Exp[S_T] 
\Exp[\beta_T] + 0.8\sigma t $, or with probability at least $\rho - 1/n^2$, 
%we have 
$\beta_T \le 
%\Exp[S_T] 
\Exp[\beta_T] - 0.8\sigma t $.
%(where $\E[S_T]$ is the empirical expectation of $S_T$).
\end{Def}

Starting with a correlated set (a potential signature set), we expand it similar to (Definition~\ref{def:expand}), except that we find $\tilT$ as follows: 

\[ \tilde{T}_{temp} = \{\textrm{$2d$ coordinates of largest {\em magnitude} in $\hat{A}_T$} \}, \tilT_1 = \{i\in \tilT_{temp}: \hat{A}_T \ge 0\}\]
\[ \tilde{T} = \left\{\begin{array}{cc} \tilT_1 & \textrm{ if } |T_1|\ge d \\ \tilT_{temp}\setminus \tilT_1 & \textrm{otherwise}\end{array}\right.\]
%\[ \tilde{T} = \{\textrm{$d$ coordinates of largest {\em magnitude} in $\hat{A}_T$}\}. \]

Our earlier definitions of expanded signature sets and bias can also be adapted naturally:
\begin{Def}[(General) Expanded Signature Set]
An expanded set $\tilde{T}$ is an expanded signature set for $x_j$ if $|\beta_{j,\tilde{T}}| \ge 0.7\sigma d$ and for all $k\ne j$, $|\beta_{k,\tilde{T}}| \le 0.3\sigma d$.
\end{Def}

Since Lemma~\ref{lem:extend} still holds, Lemma~\ref{lem:expandsignature} follows straightforwardly. That is, there always exists a general expanded signature set $\tilT$ that is produced by a set $T$ of size $t = O_\theta(\log n^2)$. (Note that this is why in the general case we assume that $G_{\sigma}$ has degree at least $2d$ in Assumption G1. We want to make the size of good expanded set to be $d$ instead of $d/2$ so that all the lemmas can be adapted without change of notation).  

\begin{Def}[(General) Empirical Bias] The empirical bias $\bias_{\tilde{T}}$ of an expanded set $\tilde{T}$ of size $d$ is defined to be the largest $B$ that satisfies
\[ \card{\set{k\in [p]:~\left| \beta^k_{\tilde{T}}- \EExp[\beta_{\tilde{T}}]\right| \ge B}} \ge \rho N/2.\]

In other words, $\bias_{\tilde{T}}$ is the difference between the $\rho N/2$-th largest $\beta^k_{\tilde{T}}$ in the samples and $\EExp[\beta_{\tilde{T}}]$.
\end{Def}

Let us now intuitively describe why the analog of Lemma~\ref{lem:maxbias} holds in the general case. We provides the formal statement and the proof in Appendix~\ref{sec:app:symmetric_diff}
\begin{enumerate}
\item The first step, Claim~\ref{claim:largeedges} is a statement purely about the magnitudes of the edges (in fact, cancellations in $\beta_{k, \tilde{T}}$ for $k \ne j$ only help our case).
\item The second step, Claim~\ref{claim:biasmaxclose} essentially argues that the small $\beta_{k, \tilde{T}}$ do not contribute much to the bias (a concentration bound, which still holds due to Lemma~\ref{lem:variance-bound-general}), and that the probability of {\em two} ``large'' features $j, j'$ being on simultaneously is very small. The latter holds even if the $\beta_{j,\tilde{T}}$  have different signs.
\item The final step in the proof of Lemma~\ref{lem:maxbias} is an argument which uses the assumption on the overlap between features to contradict the maximality of bias, when the case where $\beta_{j,\tilde{T}}$ and $\beta_{j',\tilde{T}}$ are both ``large''. This only uses the magnitudes of the entries in $A$, and thus also follows.
\end{enumerate}

\paragraph{Recovering an equivalent dictionary.}  The main lemma in the nonnegative case, which shows that Algorithm~\ref{alg:expand} {\em roughly} recovers a column, is Lemma~\ref{lem:expandexpandedset}. The proof uses the property that signature sets are elevated ``almost iff'' the $x_j=1$ to conclude that we get a good approximation to one of the columns. We have seen that this also holds in the general case, and since the rest of the argument deals only with the magnitudes of the entries, we conclude that we can roughly recover a column also in the general case.  Let us state this formally.
\begin{Lem}
\label{lem:expandexpandedset:gen}
If $\tilde{T}$ is an expanded signature set for $x_j$, and $\tilde{A}_{\tilde{T}}$ is the corresponding column output by Algorithm~\ref{alg:expand}, then with high probability $\|\tilde{A}_{\tilde{T}} - A_j\|_\infty \le O(\rho(\Lambda^3\log n/\sigma^2)^2\sqrt{\Lambda\log n}) = o(\sigma)$.
\end{Lem}

Once we have all the entries which are $> \sigma/2$ in magnitude, we can use the `refinement' trick of Lemma~\ref{lem:refinement} to conclude that we can recover the entries.

\begin{Lem}
When the number of samples is at least $n^{4C+3}m$, the matrices $A$ and $\hat{A}$ are entry-wise $n^{-2C}m^{-1/2}$ close. Further, the two dictionaries are $n^{-C}$-equivalent.
\end{Lem}

The first part of the proof (showing entry-wise closeness) is very similar to Lemma~\ref{lem:extend}. In order to show $n^{-C}$ equivalent, notice when the entries are very close this just follows from Bernstein's inequality, with variance bounded by $n^{-4C}m^{-1} \cdot m$. In Section 3 we do not just use this bound, because we want to be able to also handle the case when the entrywise error is only inverse polylog (for Assumption 2).

\iffalse
The following is an analog of Lemma~\ref{lem:equivalent}.

\begin{Lem}
Let $\sigma = f(\epsilon/\Delta\sqrt{\log n})$ (where $\Delta$ is a large enough constant), let $\hat{A}$ be a matrix in which the entries are either equal to 0 or have absolute value at least $\sigma/2$, for all entries $|\mat{A}{i}{j}| \ge \sigma$, $|\mat{A}{i}{j} - \mat{\hat{A}}{i}{j}| \le \epsilon\sigma/3\Delta\sqrt{\Gamma\log n}$, then $A$ and $\hat{A}$ are $\epsilon$-equivalent.
\end{Lem}

\begin{proof}
First consider a matrix $A'$ which is $A$ restricted to entries that are nonzero in $\hat{A}$. By Assumption 5, and Bernstein's inequality we know $A$ and $A'$ are $\epsilon/2$-equivalent (when $\Delta$ is large enough).

On the other hand, the $\ell_2$ norm of $\row{\hat{A}}{i} - \row{A'}{i}$ is smaller than $\epsilon/\Delta\sqrt{\rho \log n}$, because every entry of $\row{\hat{A}}{i} - \row{A'}{i}$ has absolute value within $\epsilon/\Delta\sqrt{\Gamma\log n}$ times the corresponding entry in $A'$, and $A'$ has $\ell_2$ norm bounded by $\sqrt{\Gamma/\rho}$.  So $(\row{\hat{A}}{i} - \row{A'}{i})x \le \epsilon/2$ with high probability. Hence $A'$ and $\hat{A}$ are $\epsilon/2$-equivalent.

It is not hard to show that $\epsilon$-equivalent satisfies triangle inequality. Therefore $A$ and $\hat{A}$ are $\epsilon$-equivalent.
\end{proof}
\fi
\bibliography{ref}

\newcommand{\etalchar}[1]{$^{#1}$}
\begin{thebibliography}{YWHM08}

\bibitem[AAN13]{DBLP:journals/corr/AgarwalAN13}
Alekh Agarwal, Animashree Anandkumar, and Praneeth Netrapalli.
\newblock Exact recovery of sparsely used overcomplete dictionaries.
\newblock {\em CoRR}, abs/1309.1952, 2013.

\bibitem[ABGM13]{DBLP:journals/corr/AroraBGM13}
Sanjeev Arora, Aditya Bhaskara, Rong Ge, and Tengyu Ma.
\newblock Provable bounds for learning some deep representations.
\newblock {\em CoRR}, abs/1310.6343, 2013.

\bibitem[AEB05]{aharon2005k}
Michal Aharon, Michael Elad, and Alfred~M Bruckstein.
\newblock K-svd and its non-negative variant for dictionary design.
\newblock In {\em Optics \& Photonics 2005}, pages 591411--591411.
  International Society for Optics and Photonics, 2005.

\bibitem[AEB06]{aharon2006img}
Michal Aharon, Michael Elad, and Alfred Bruckstein.
\newblock K-svd: An algorithm for designing overcomplete dictionaries for
  sparse representation.
\newblock {\em Signal Processing, IEEE Transactions on}, 54(11):4311--4322,
  2006.

\bibitem[AEP06]{DBLP:conf/nips/ArgyriouEP06}
Andreas Argyriou, Theodoros Evgeniou, and Massimiliano Pontil.
\newblock Multi-task feature learning.
\newblock In {\em NIPS}, pages 41--48, 2006.

\bibitem[AGM13]{AGM}
Sanjeev Arora, Rong Ge, and Ankur Moitra.
\newblock New algorithms for learning incoherent and overcomplete dictionaries.
\newblock {\em ArXiv}, 1308.6273, 2013.

\bibitem[Aha06]{AharonThesis}
Michal Aharon.
\newblock {\em Overcomplete Dictionaries for Sparse Representation of Signals}.
\newblock PhD thesis, Technion - Israel Institute of Technology, 2006.

\bibitem[BC{\etalchar{+}}07]{boureau2007sparse}
Y-lan Boureau, Yann~L Cun, et~al.
\newblock Sparse feature learning for deep belief networks.
\newblock In {\em Advances in neural information processing systems}, pages
  1185--1192, 2007.

\bibitem[Ben62]{Bernnett62}
George Bennett.
\newblock Probability inequalities for the sum of independent random variables.
\newblock {\em Journal of the American Statistical Association}, 57(297):pp.
  33--45, 1962.

\bibitem[Ber27]{Bernstein}
S.~Bernstein.
\newblock {\em Theory of Probability}, 1927.

\bibitem[BGI{\etalchar{+}}08]{IndykStrauss}
R.~Berinde, A.C. Gilbert, P.~Indyk, H.~Karloff, and M.J. Strauss.
\newblock Combining geometry and combinatorics: a unified approach to sparse
  signal recovery.
\newblock In {\em 46th Annual Allerton Conference on Communication, Control,
  and Computing}, pages 798--805, 2008.

\bibitem[CRT06]{candes2006robust}
Emmanuel~J Cand{\`e}s, Justin Romberg, and Terence Tao.
\newblock Robust uncertainty principles: Exact signal reconstruction from
  highly incomplete frequency information.
\newblock {\em Information Theory, IEEE Transactions on}, 52(2):489--509, 2006.

\bibitem[Das99]{DBLP:conf/focs/Dasgupta99}
Sanjoy Dasgupta.
\newblock Learning mixtures of gaussians.
\newblock In {\em FOCS}, pages 634--644. IEEE Computer Society, 1999.

\bibitem[DH01]{donoho2001uncertainty}
David~L Donoho and Xiaoming Huo.
\newblock Uncertainty principles and ideal atomic decomposition.
\newblock {\em Information Theory, IEEE Transactions on}, 47(7):2845--2862,
  2001.

\bibitem[DMA97]{davis1997adaptive}
Geoff Davis, Stephane Mallat, and Marco Avellaneda.
\newblock Adaptive greedy approximations.
\newblock {\em Constructive approximation}, 13(1):57--98, 1997.

\bibitem[EA06]{elad2006image}
Michael Elad and Michal Aharon.
\newblock Image denoising via sparse and redundant representations over learned
  dictionaries.
\newblock {\em Image Processing, IEEE Transactions on}, 15(12):3736--3745,
  2006.

\bibitem[EAHH99]{engan1999method}
Kjersti Engan, Sven~Ole Aase, and J~Hakon~Husoy.
\newblock Method of optimal directions for frame design.
\newblock In {\em Acoustics, Speech, and Signal Processing, 1999. Proceedings.,
  1999 IEEE International Conference on}, volume~5, pages 2443--2446. IEEE,
  1999.

\bibitem[Hoy02]{hoyer2002non}
Patrik~O Hoyer.
\newblock Non-negative sparse coding.
\newblock In {\em Neural Networks for Signal Processing, 2002. Proceedings of
  the 2002 12th IEEE Workshop on}, pages 557--565. IEEE, 2002.

\bibitem[Ind08]{DBLP:conf/soda/Indyk08}
Piotr Indyk.
\newblock Explicit constructions for compressed sensing of sparse signals.
\newblock In Shang-Hua Teng, editor, {\em SODA}, pages 30--33. SIAM, 2008.

\bibitem[JXHC09]{DBLP:journals/tit/JafarpourXHC09}
Sina Jafarpour, Weiyu Xu, Babak Hassibi, and A.~Robert Calderbank.
\newblock Efficient and robust compressed sensing using optimized expander
  graphs.
\newblock {\em IEEE Transactions on Information Theory}, 55(9):4299--4308,
  2009.

\bibitem[LS99]{lee1999learning}
Daniel~D Lee and H~Sebastian Seung.
\newblock Learning the parts of objects by non-negative matrix factorization.
\newblock {\em Nature}, 401(6755):788--791, 1999.

\bibitem[LS00]{lewicki2000learning}
Michael~S Lewicki and Terrence~J Sejnowski.
\newblock Learning overcomplete representations.
\newblock {\em Neural computation}, 12(2):337--365, 2000.

\bibitem[MLB{\etalchar{+}}08]{mairal2008discriminative}
Julien Mairal, Marius Leordeanu, Francis Bach, Martial Hebert, and Jean Ponce.
\newblock Discriminative sparse image models for class-specific edge detection
  and image interpretation.
\newblock In {\em Computer Vision--ECCV 2008}, pages 43--56. Springer, 2008.

\bibitem[OF97]{olshausen1997sparse}
Bruno~A Olshausen and David~J Field.
\newblock Sparse coding with an overcomplete basis set: A strategy employed by
  v1?
\newblock {\em Vision research}, 37(23):3311--3325, 1997.

\bibitem[SWW12]{DBLP:journals/jmlr/SpielmanWW12}
Daniel~A. Spielman, Huan Wang, and John Wright.
\newblock Exact recovery of sparsely-used dictionaries.
\newblock {\em Journal of Machine Learning Research - Proceedings Track},
  23:37.1--37.18, 2012.

\bibitem[YWHM08]{yang2008image}
Jianchao Yang, John Wright, Thomas Huang, and Yi~Ma.
\newblock Image super-resolution as sparse representation of raw image patches.
\newblock In {\em Computer Vision and Pattern Recognition, 2008. CVPR 2008.
  IEEE Conference on}, pages 1--8. IEEE, 2008.

\end{thebibliography}
\bibliographystyle{alpha}

\appendix
\section{Full Proofs}

In this section we give the omitted proofs.

\subsection{Proof of Lemma~\ref{lem:finalset}}
\label{sec:lem:finalset}
\begin{proof}
Let us focus on the $i$th row of $A-\hat{A}$ and denote it by $w$.  Then we have $\norm{w}_1 \le \norm{A}_1+ \norm{\hat{A}}_1 \le O(1/\rho)$.  Now consider the random variable $Z = \sum_j w_j x_j$, where $x_j$ are i.i.d. Bernoulli r.v.s with probability $\rho$ of being $1$. Then by Bernstein's inequality (Theorem~\ref{thm:bernstein_ineq}), we have
\[ \Pr[ Z - \E Z > \epsilon ] \le e^{-\frac{\epsilon^2}{\rho \sigma + \rho \sum_j w_j^2}}. \]
Since $|w_j| < \delta$ for all $j$, we can bound the variance as
$\rho \cdot \sum_j w_j^2 \le \delta \rho \cdot \sum_j |w_j| \le 2\delta$. 

Thus setting $t= (4\delta \log n)^{1/2}$ (notice that this is the $t$ in Bernstein's inequality, not the same as the size of signature sets), we obtain an upper bound of $1/\text{poly}(n)$ on the probability.
\end{proof}

\subsection{Missing Lemmas and Proofs of Section~\ref{sec:general}}\label{sec:app:symmetric_diff}

\begin{Lem}[General Version of Lemma~\ref{lem:symmetric_diff}]\label{lem:gen_symmetric_diff}
Suppose $T$ of size $t$ is a general signature set for $x_j$ with $t =\omega(\sqrt{\log n})$. Let $E_1$ be the event that $x_j =1$ and $E_2$ be the event that $\beta_T \ge \Exp[\beta_T]+ 0.9\sigma t$ if $\beta_{j,T} \ge \sigma t$, and the event $\beta_T \le \Exp[\beta_T] - 0.9\sigma t$ if $\beta_{j,T} \le - \sigma t$.  Then for large constant $C$ (depending on $\Delta$) 

%\footnote{$\Delta$ was introduced in Definition~\ref{def:signature_set}})
\begin{enumerate}
\item $\Pr[E_1] + n^{-2C}\ge \Pr[E_2] \ge \Pr[E_1] - n^{-2C}$. 
\item $\Pr[E_2|E_1] \ge 1 - n^{-2C}$, and $\Pr[E_2|E_1^c] \le n^{-2C}$.
\item $\Pr[E_1|E_2]\ge 1 - n^{-C}$.
\end{enumerate} 
%\Tnote{or $n^{-1.5}/\rho$ , should check constant}
\end{Lem}
\begin{proof}
It is a straightforward modification of the proof of Lemma~\ref{lem:symmetric_diff}. First of all, $|\Exp[\beta_{j,T}x_j]| = o(\sigma t)$,and thus mean of $\sum_{k\neq j} \beta_{k,T}x_k$ only differs from that of $\beta_{T}$ by at most $o(\sigma t)$. Secondly, Bernstein inequality requires the largest coefficients and the total variance be bounded, which correspond to exactly property (b) and (c) of a general signature set. The rest of the proof follows as in that for Lemma~\ref{lem:symmetric_diff}. 
\end{proof}

\begin{Lem}[General version of Lemma~\ref{lem:maxbias}]\label{lem:gen_maxbias}
Let $\tilde{T}^*$ be the set with largest {\em general} empirical bias $\bias_{\tilde{T}^*}$ among all the expanded sets $\tilde{T}$. The set $\tilde{T}^*$ is an expanded signature set for some $x_j$.
\end{Lem}
\newcommand{\thr}{\theta}
\newcommand{\KLL}{\KL}
%-----------------------ANOTHER REQUIREMENT FOR \gamma----------------
%
%\Tnote{Need $\gamma \le \sfrac{\sigma^4}{\Delta\Lambda^2\log n}$}
%
%
%-----------------------ANOTHER REQUIREMENT FOR \gamma----------------

\begin{proof}
We first prove an analog of Claim~\ref{claim:largeedges}. Let $W = \sfrac{\sigma^4d}{2\Delta\Lambda^3\log n}$. Let's redefine $\KLL := \{k\in [m]: \card{\{i \in \tilT: |\mat{A}{i}{k}| \ge \tau \}} \ge W\}$ be the subset of nodes in $[m]$ which connect to at least $W$ nodes in $\tilT$ in the subgraph $G_{\tau}$.Note that this implies that if $k\not\in \KLL$, then $|\beta_{k,\tilT}|\le d\tau + W\Lambda\le \sfrac{d\sigma^4}{\Delta\Lambda^2\log n}$. Let $Q_k =  \{i\in \tilT: |\mat{A}{i}{k}| \ge \tau\}$. By definition, we have for $k\in \KLL$, $|Q_K|\ge W$. Then similarly as in the proof of Claim~\ref{claim:largeedges}, using the fact that $|Q_{k}\cap Q_{k'}|\le \kappa$, and inclusion-exclusion, we have that $|\KLL|\le O(\Delta\Lambda^3\log n/\sigma^4)$.

Then we prove an analog of Claim~\ref{claim:biasmaxclose}. Let $\beta_{small,\tilT}$ and $\beta_{large,\tilT}$ be defined as in the proof of Claim~\ref{claim:biasmaxclose} (with the new definition of $\KLL$). By Lemma~\ref{lem:variance-bound-general}, the variance of $\beta_{small,\tilT}$ is bounded by $2dW+2d^2\gamma\le 2\sfrac{d^2\sigma^4}{\Delta\Lambda^2\log n}$. Therefore by Bernstein's inequality we have that for sufficiently large $\Delta$, with probability at least $1-n^{-2}$ over the choice of $x$, $|\beta_{small, \tilde{T}} - \E[\beta_{small, \tilde{T}}]|\le 0.05d\sigma^2/\Lambda$. It follows from the same argument of Claim~\ref{claim:biasmaxclose} that with high probability over the choice of $N$ samples, $|\bias_T - \max_k\beta_{k,\tilT}|\le 0.1d\sigma^2/\Lambda$ holds when $\max_k\beta_{k,\tilT} \ge 0.5d\sigma$.

We apply almost the same argument as in the proof of Lemma~\ref{lem:maxbias}. By Lemma~\ref{lem:expandexpandedset:gen} we know that our algorithm must produce an expanded signature set of size $d$ with bias at least $0.8\sigma d$, and thus the set $\tilT^*$ with largest bias must has a large coefficient $j$ with $\beta_{j,\tilT^*}\ge 0.7\sigma d$. If there is some other $k$ such that $\beta_{k,\tilT^*}\ge 0.3\sigma d$, then $|Q_k|\ge 0.3\sigma d/\Lambda$ and therefore we could remove those elements in $\tilT^* - Q_j$, which has size larger than $0.3\sigma d/\Lambda-\kappa$ by Assumption G2. Then by adding some other elements which are in the neighborhood of $j$ in $G_{\sigma}$ into the set $Q_j$ we get a set with bias larger than $\tilT^*$, which contradicts our assumption that there exists $k$ with $\beta_{k,\tilT^*} \ge 0.3\sigma d$. Hence $\tilT^*$ is indeed an expanded signature set and the proof is complete.

%We inherent most of the notations ... 
%Let $K = \{k: \beta_{k,\tilT} \ge \thr\}$ be the set of large $\beta_{k,\tilT}$'s, $\beta_{small,\tilT} = \sum_{k \in K-k^*}\beta_{k,\tilT}x_k$ and $\beta_{small,\tilT} = \sum_{k\not\in K}\beta_{k,\tilT}x_k$. We would like to use to Bernstein's inequality to show $\beta_{small,\tilT}$ concentrates around its mean. Each $\beta_{k,\tilT}$ has upper bound $M = \thr$, and there for the number of edges from $x_k$ to $T$ in graph $G_{\delta}$ is also upper bounded by $\thr/\delta$. Using \ref{lem:variance-bound-general}, we have that variance of $\sum \beta_{k,\tilT}x_k$ is bounded by $d \thr/\delta + d^2\gamma^2 \le d^2\sigma^4/\log n.$
\end{proof}

%\Tnote{To make this work, we need: 1. $\theta/\delta \le \sigma^4/\log n$ 2. $\gamma^2 \le \sigma^4/\log n$. to make $\gamma^2 = \sigma^4/\log n$, it is natural to take $\delta = \sigma^4/\log n$ and thus $\theta = \sigma^8/\log^2 n$. Since we need that the number of $\beta_{k,\tilT}$ that is larger than $\theta$ is smaller than $\poly \log$, we need to reprove Claim 9 with $\kappa  = \sigma^O(1)/\log^3 n$ }. 

\section{Probability Inequalities}
\begin{Lem}\label{lem:prob_ineq}
Suppose $X$ is a bounded random variable in a normed vector space with $||X|| \le M$. If event $E$ happens with probability $1-\delta$ for some $\delta < 1$, then 
$||\Exp[X|E] - \Exp[X]|| \le 2\delta M$
\end{Lem}

\begin{proof}
We have $\Exp[X] = \Exp[X | E]\Pr[E] + \Exp[X | E^c]\Pr[E^c] = \Exp[X|E] + (\Exp[X | E^c]- \Exp[X | E])\Pr[E^c]$, and therefore  $ ||\Exp[X|E] - \Exp[X]|| \le 2\delta M$. 
\end{proof}

\begin{Lem}
Suppose $X$ is a bounded random variable in a normed vector space with $||X|| \le M$. If events $E_1$ and $E_2$ have small symmetrical differences in the sense that $\Pr[E_1|E_2] \le \delta$ and  $\Pr[E_2|E_1] \le \delta$. Then $||\Exp[X|E_1] - \Exp[X|E_2]||\le 4\delta M$. 
\end{Lem}

\begin{proof}
Let $Y = X|E_2$, by Lemma~\ref{lem:prob_ineq}, we have $||\Exp[Y|E_1] - \Exp[Y]||\le 2\delta M$, that is, $||\Exp[X|E_1E_2] - \Exp[X|E_2]||\le 2\delta M$. Similarly $||\Exp[X|E_1E_2] - \Exp[X|E_1]||\le 2\delta M$, and hence $||\Exp[X|E_1] - \Exp[X|E_2]||\le 4\delta M$.
\end{proof}

\begin{theorem}[Bernstein Inequality\cite{Bernstein} cf.~\cite{Bernnett62}]\label{thm:bernstein_ineq}
Let $x_1,\dots, x_n$ be independent variables with finite variance $\sigma_i^2 = \Var[x_i]$ and bounded by $M$ so that $|x_i - \Exp[x_i]|\le M$. Let $\sigma^2 = \sum_i \sigma_i^2$. Then we have
\[\Pr\left[\left|\sum_{i=1}^n x_i - \Exp[\sum_{i=1}^n x_i] \right| > t\right] \le 2\exp(- \frac{t^2}{2\sigma^2 + \frac{2}{3} M t})\]
\end{theorem}

\end{document}